\documentclass[11pt]{article}

\usepackage{amssymb,amsmath,amsfonts}
\usepackage{graphicx,color,enumitem}
\usepackage{amsthm} 
\usepackage{bm}
\usepackage{geometry}

\usepackage[colorinlistoftodos, textwidth=4cm, shadow]{todonotes}
\RequirePackage[colorlinks,citecolor=blue,urlcolor=blue]{hyperref}

\newcommand{\E}{{\mathbb E}}
\newcommand{\F}{{\mathbb F}}

\renewcommand{\P}{{\mathbb P}}

\newcommand{\C}{{\mathbb C}}
\newcommand{\R}{{\mathbb R}}

\newcommand{\Acal}{{\mathcal A}}
\newcommand{\Bcal}{{\mathcal B}}
\newcommand{\Ccal}{{\mathcal C}}
\newcommand{\Dcal}{{\mathcal D}}

\newcommand{\Fcal}{{\mathcal F}}
\newcommand{\Gcal}{{\mathcal G}}
\newcommand{\Hcal}{{\mathcal H}}

\newcommand{\Kcal}{{\mathcal K}}
\newcommand{\Lcal}{{\mathcal L}}

\DeclareMathOperator{\argmin}{arg \, min}

\DeclareMathOperator{\esssup}{ess\,sup}

\newtheorem{theorem}{Theorem}

\newtheorem{assumption}[theorem]{Assumption}

\newtheorem{definition}[theorem]{Definition}

\newtheorem{lemma}[theorem]{Lemma}

\newtheorem{proposition}[theorem]{Proposition}
\newtheorem{remark}[theorem]{Remark}

\theoremstyle{definition}
\newtheorem{example}[theorem]{Example}

\numberwithin{equation}{section}
\numberwithin{theorem}{section}

\definecolor{darkgreen}{rgb}{0,0.7,0}

\newcommand{\iii}{{\vert\kern-0.25ex\vert\kern-0.25ex\vert}}

\newcommand{\dd}[1]{\operatorname{d}\!#1}

\def \eps{\varepsilon}

\title{American options in the Volterra Heston model\footnote{The authors would like to thank Eduardo Abi Jaber for useful comments and fruitful discussions. The research of Sergio Pulido benefited from the financial support of the chairs ``Deep finance \& Statistics" and ``Machine Learning \& systematic methods in finance" of \'Ecole Polytechnique. Sergio Pulido acknowledges support by the Europlace Institute of Finance (EIF) and the Labex Louis Bachelier, research project: ``The impact of information on financial markets". The work of Elizabeth Z\'u\~niga was supported by grants from the Fondation Math\'ematique Jacques Hadamard (FMJH).}}
\author{Etienne Chevalier\thanks{Universit\'e Paris-Saclay, CNRS, Univ Evry, Laboratoire de Math\'ematiques et Mod\'elisation d'Evry (LaMME), etienne.chevalier@univ-evry.fr} \and Sergio Pulido\thanks{Universit\'e Paris-Saclay, CNRS, ENSIIE, Univ \'Evry, Laboratoire de Math\'ematiques et Mod\'elisation d'Evry (LaMME), sergio.pulidonino@ensiie.fr} \and Elizabeth Z\'u\~niga\thanks{Universit\'e Paris-Saclay, CNRS, Univ Evry, Laboratoire de Math\'ematiques et Mod\'elisation d'Evry (LaMME), elizabeth.zunigariofrio@univ-evry.fr} }

\begin{document}

\maketitle

\begin{abstract}
We price American options using kernel-based approximations of the Volterra Heston model. We choose these approximations because they allow simulation-based techniques for pricing. We prove the convergence of American option prices in the approximating sequence of models towards the prices in the Volterra Heston model. A crucial step in the proof is to exploit the affine structure of the model in order to establish explicit formulas and convergence results for the conditional Fourier--Laplace transform of the log price and an adjusted version of the forward variance. We illustrate with numerical examples our convergence result and the behavior of American option prices with respect to certain parameters of the model.

\end{abstract}


\section{Introduction}

Stochastic volatility models whose trajectories are continuous but less regular than Brownian motion, also known as rough volatility models, seem well adapted to capture stylized features of the time series of realized volatility and of the implied volatility surface. Indeed, recent statistical studies in \cite{gatheral2018volatility, fukasawa2019volatility,bennedsen2016decoupling} demonstrate that---under multiple time scales and across many markets---the time series of realized volatility oscillates more rapidly than Brownian motion. In addition, the observed implied volatility smile for short maturities is steeper than the one obtained with classical low-dimensional diffusion models. As maturity decreases, the slope at the money of the implied volatility smile obeys a power law that explodes at zero. This power law can be reproduced by rough volatility models with power kernels in the spirit of fractional Brownian motion; cf. \cite{alos2007short,bayer2016pricing,fukasawa2017short}. Furthermore, these empirical discoveries are supported by microstructural considerations because, as explained, for instance, in  \cite{jaisson2016rough,el2018microstructural}, rough volatility models appear naturally as scaling limits of microstructural pricing models with self-exciting features driven by Hawkes processes.

The aforementioned findings have motivated the study of various rough volatility models in the literature. Among these are the rough fractional stochastic volatility model \cite{gatheral2018volatility}, the rough Bergomi model \cite{bayer2016pricing}, and the fractional and rough Heston models \cite{comte2012affine,guennoun2018asymptotic,el2019characteristic}. In these models, due to the absence of the semimartingale and Markov properties, even simple tasks such as pricing European options have proven challenging. Consequently, the theory of stochastic control for rough volatility models is at an early stage. Under the rough volatility paradigm, classical control problems such as linear quadratic and optimal investment problems have only been analyzed recently, for example, in \cite{jaber2019linear} and \cite{fouque2019optimal, han2020mean, abijaberMarkowitz2020, han2021merton}, respectively. 

In this paper we tackle an optimal stopping problem, namely the problem of pricing American options in the Volterra Heston model introduced in \cite{abi2019markovian,abi2019affine}\footnote{This is an important problem as the majority of equity options are of American type. This can be confirmed, for instance, for stock options traded on Euronext, by visiting \url{https://live.euronext.com/en/products/stock-options/list}.}.  This path-dependent problem is difficult because it requires a good understanding of the conditional laws in a model where, in general, the semimartingale and Markov properties do not hold. Even though we could extend parts of the analysis to more general frameworks, we concentrate on the Volterra Heston model because in this setup---as we will explain below---we can prove the necessary convergence results.

The Volterra Heston model is a generalization of the widely known Heston model \cite{heston1993closed}. The dynamics of the spot variance in the Volterra Heston model are described by a stochastic Volterra equation of convolution type. More specifically, the spot variance process is a Volterra square root or CIR process. When the kernel appearing in the convolution is of power-type, one obtains the now well-known rough Heston model \cite{el2018perfect,el2019characteristic}. The $\Lcal^2$-regularity of the kernel in the Volterra Heston model controls the H\"older regularity of the trajectories and the steepness of the implied volatility smile for short maturities. Tractability in the Volterra Heston model is a result of a semiexplicit formula for the Fourier--Laplace transform, which resembles the formula in the classical Heston model. More precisely, the Fourier--Laplace transform can be expressed in terms of the solution to a deterministic system of convolution equations of Riccati-type. This phenomenon is a particular instance of a more general law governing the structure of the Fourier--Laplace transform of what is known as Affine Volterra Processes \cite{keller2018affine, abi2019affine, gatheral2019affine, cuchiero2020generalized}. The knowledge of the Fourier--Laplace transform in the Volterra Heston model facilitates the application of Fourier-based methods in order to price European options. This circumvents the difficulties encountered in the implementation of other popular rough volatility models, such as the rough Bergomi model, where Monte-Carlo techniques \cite{bayer2016pricing, bennedsen2017hybrid, bayer2020hierarchical} or Donsker-type theorems \cite{horvath2017functional} are employed to compute prices of European options. 

The numerical resolution of the Riccati convolution equations appearing in the expression of the Fourier--Laplace transform in the Volterra Heston model is, however, cumbersome due to the possibly exploding character of the associated kernel. In order to alleviate these numerical difficulties, the authors in \cite{callegaro2021fast} developed a fast hybrid scheme using power series expansions around zero combined with a Richardson--Romberg extrapolation method beyond the convergence interval of the power series.  Alternatively, for the rough Heston model, the author in \cite{abi2019lifting} proposed a kernel-based approximation  with a diffusion---high dimensional but parsimonious---model, named the Lifted Heston model. Despite being a semimartingale model, the Lifted Heston model is able to mimic the rough character of the trajectories and to reproduce steep volatility smiles for short maturities. The approximation of the rough Heston model with the Lifted Heston model is an example of a more general approximation technique of Volterra processes via an approximation of the kernel in \cite{abi2019multifactor} originally inspired by \cite{coutin1998fractional,carmona2000approximation, harms2019affine}. The convergence of the approximating processes and the prices of European-type options is guaranteed by stability results proven in \cite{abi2019multifactor} and in a more general framework in \cite{jaber2019weak}.

To price American options, and inspired by the approach in \cite{abi2019lifting, abi2019multifactor}, we draw upon kernel-based approximations of the Volterra Heston model. In the context of the rough Heston model where the kernel is of power-type, and for the approximation scheme in \cite{abi2019lifting}, the approximating models are high dimensional-diffusion models where classical simulation-based techniques, such as the Longstaff Schwartz algorithm \cite{longstaff2001valuing}, can be implemented. Within this framework, we can conduct an empirical study of the convergence and behavior of Bermudan put option prices in the approximated sequence of models. The results of our numerical experiments are summarized in Section \ref{sec:numeric}. 

Our main theoretical result is Theorem \ref{thm:main}. In the first part of the theorem, we show convergence of prices of Bermudan options in the approximating sequence of models towards the prices in the original Volterra Heston model. This result is not a direct consequence of previous stability results in \cite{abi2019multifactor,jaber2019weak} because of the path-dependent structure of the option. It is at this stage, and for purely theoretical reasons, that we exploit the affine structure of the model. More precisely, in order to prove the desired convergence results we first need to establish the convergence of the conditional Fourier--Laplace transforms. Once the convergence of the Bermudan option prices is established---and using classical arguments---we can prove, in the second part of Theorem \ref{thm:main}, the convergence of American option prices by approximating them with Bermudan option prices.

It is important to mention at this point that there exist other studies of optimal stopping and American option pricing in rough or fractional models; see, for instance, \cite{horvath2017functional, becker2019deep, bayer2020pricing, goudenege2020machine}. To understand the novelty of our work it is crucial to point out that, in general, there are two levels of approximation in the resolution of an optimal stopping problem using a probabilistic approach: 
\begin{itemize}
\item[(i)] First, the model has to be  approximated with simpler models where the trajectories can be simulated or where prices of American options can be computed more easily. For classical diffusion models this could correspond to a classical Euler scheme for simulation or a tree-based discrete approximation. Under rough volatility, simulation is cumbersome due to the non-Markovianity of the model. There is not a unified theory about how this approximation and simulation have to be performed. For instance, in the rough Bergomi model in order to simulate the volatility process one could use hybrid schemes \cite{bennedsen2017hybrid}. These schemes correspond to an approximation of the power kernel by concentrating on its behavior around zero and performing a stepwise approximation away from zero. But we could also imagine schemes relying on an approximation of the fractional kernel in terms of a sum of exponentials as in \cite{carmona2000approximation, harms2019affine}. Other recent studies in this direction are \cite{bayer2020weak, alfonsi2021approximation, bayer2021makovian, romer2020hybrid}. In this work, for our numerical illustrations, we use the approximation scheme of \cite{abi2019lifting,abi2019multifactor}, based on an approximation of the kernel using a sum of exponentials. Regarding the approximation via discrete-type models, in \cite{horvath2017functional} the authors prove a Donsker-type theorem for certain rough volatility models and apply it to perform tree-like approximations. These approximations allow them to develop tree-based algorithms, as opposed to simulation-based techniques, to price American options. The convergence of the American option prices computed on the approximating trees towards prices in the limiting rough models, however, is not the main goal of the study. 
\item[(ii) ] The second approximation occurs at the level of the resolution of the optimal stopping problem for the approximated model. In the approximated model, classical techniques such as the Longstaff Schwarz algorithm can be difficult to implement because of the high dimensionality of the model. It is at this stage that recent studies propose novel approaches, including techniques relying on neural networks \cite{goudenege2020machine, lapeyre2019neural}, to ease the implementation. It is also important to mention at this point the study in \cite{bayer2020pricing2}, where the authors propose an approximation of American option prices using penalized versions of the backward stochastic partial differential equation (BSPDE) satisfied by the value function of the problem. A deep learning-based method is used to approximate the solutions of these penalized BSPDEs.
\end{itemize}
The present paper does not focus on the second level of the approximation. For this part, in our numerical experiments we employ classical simulation-based techniques and, in particular, the Longstaff Schwarz algorithm over a low-dimensional space of functions. Our study mainly focuses on the first level of the approximation. More precisely, we concentrate on the convergence of the prices in the approximating model towards the prices in the limiting Volterra model. This point has not been addressed in the previous literature and is what distinguishes our paper from other papers on American options under the rough volatility paradigm. To prove this convergence in our framework and with our kernel-based approximation approach, we appeal to the particular affine structure of the Volterra Heston model, which explains our choice of setting. One could extend some of the results to other settings as long as the results regarding the convergence of the conditional Fourier--Laplace transform remain valid. Beyond the affine paradigm, for instance for the rough Bergomi model, this question falls outside the scope of our work and is an interesting topic for future research.

The rest of this paper is organized as follows. In Section \ref{sec:setup} we introduce the setup and state our main result of convergence, namely Theorem \ref{thm:main}. Section \ref{sec:conditional} contains the results on the adjusted forward process and the conditional Fourier--Laplace transform necessary for the proof of the main theorem. The proof of the main theorem is presented in Section \ref{sec:proof}. In Section \ref{sec:numeric}, within the framework of the rough Heston model, we provide numerical illustrations of the convergence and behavior of Bermudan put option prices. Appendix \ref{sec:app:ricatti} explains some properties of the Riccati equations appearing in the expression of the conditional Fourier--Laplace transform. In Appendix \ref{sec:app:kernelapprox} we provide results on the kernel approximation which guarantee certain hypotheses appearing in our main theorem.

\medskip

\noindent{\bf Notation}

\medskip 

\noindent We denote by $\Lcal^2_{loc}$ the space of real-valued locally square integrable functions on $\R_+$.  Similarly, given $T>0$, $\Lcal^2(0,T)$ stands for the space of real-valued square integrable functions on the interval $(0,T)$. The space $\Ccal(X,Y)$, where $X,Y\subseteq \C$, is the space of continuous functions from $X$ to $Y$, with the conventions $\Ccal(X,\R)=\Ccal(X)$ and $\Ccal=\Ccal(\R_+)$. We use the same conventions for $\Ccal_b$, $\Ccal_b^2$, $\Ccal_c$, $\Hcal^{\beta}$, $\Bcal$, and $\Bcal_c$, which are the spaces of bounded continuous functions, bounded continuous functions with bounded and continuous derivatives up to order two, continuous functions with compact support, H\"older continuous functions of any order less than $\beta$, bounded functions, and bounded functions with compact support, respectively. We write $\Delta$ for the shift operator, i.e., $\Delta_{\xi}f=f(\cdot+\xi)$. For a function $h$ on $\R$ we denote its support by ${\rm supp}(h)$. Given a function $K$ and a measure $L$ of locally bounded variation, we let $K*L$ be the convolution $(K\ast L)(t)=\int_{[0,t]}K(t-s)L(\dd s)$, whenever the integral is well defined.  If $F$ is a function on $\R_+$, we define $K\ast F=K\ast (F\dd s)$.

\section{Setup and main result}
\label{sec:setup}

\subsection{The model}

We consider a Volterra Heston stochastic volatility model as in \cite{abi2019markovian,abi2019affine}. In this model, under a risk-neutral measure, the asset's log price $X$ and spot variance $V$ are 
\begin{equation}
\label{defXV}
\begin{aligned}
X_t&=X_0+\int_0^t \left(r-\frac{V_s}{2}\right) \dd s+\int_0^t\sqrt{V_s}\left( \rho \dd W_s+\sqrt{1-\rho^2}\dd W^{\perp}_s\right),\\
V_t&=v_0(t)-\lambda \int_0^t K(t-s)V_s\dd s+\eta\int_0^t K(t-s)\sqrt{V_s} \dd W_s.
\end{aligned}
\end{equation}
In these equations, $X_0\in\R$ is the initial log price, $(W,W^{\perp})$ is a two-dimensional Brownian motion, $r$ is the risk-free rate, and $\rho\in [-1,1]$ is a correlation parameter. The variance process $V$ is a Volterra square root process. The constant $\lambda\geq 0$ is a parameter of mean reversion speed and $\eta\geq 0$ is the volatility of volatility. The kernel $K$ is in $\Lcal^2_{loc}$ and the function $v_0$ is in $\Ccal$. Observe that---for fixed $X_0$, interest rate $r$, and correlation parameter $\rho$---the log price process $X$ is completely determined by the variance process $V$ and the Brownian motion $(W,W^{\perp})$. Proposition \ref{prop:existuniq} gives sufficient conditions ensuring the existence and uniqueness of weak solutions to the stochastic Volterra equation of the variance process. 

Following the setting in \cite{abi2019affine}, we introduce a subset $\mathcal K$ of $\Lcal^{2}_{loc}$ in which we will consider the kernels.
\begin{definition}
\label{def:Kcal}
Let $K\in \Lcal^2_{loc}$. We write $K\in\mathcal K$ if the following holds:
\begin{enumerate}
\item\label{def:Kcal:1} There exist a constant $\gamma\in (0,2]$ and a locally bounded function $c_{K}:\R_+\to\R_+$ such that 
\begin{equation}\label{eq:Kcal_c}
\int_{0}^{\xi} |K(t)|^2 \dd t + \int_{0}^{T-\xi}  |K(t+\xi)-K(t)|^2 \dd t  \leq  c_K(T) \xi^\gamma
\end{equation}
for every $T>0$ and $0< \xi\leq T$.
\item \label{def:Kcal:2} $K$ is nonidentically zero, nonnegative, nonincreasing, continuous on $(0,\infty)$, and admits a so-called resolvent of first kind $L$.\footnote{This is a real-valued measure $L$ of locally bounded variation on $\mathbb{R_+}$ such that $K \ast L= 1.$} In addition, $L$ is nonnegative and
$$\textrm{the function }s\mapsto L([s,s+t])\textrm{ is nonincreasing on }\R_+$$
for every $t>0$.
\end{enumerate}
\end{definition}
Inspired by \cite{abi2019markovian}, we specify the space of functions in which we will take the functions $v_0$. For a given kernel $K\in\Kcal$, with associated constant $\gamma$ as in \eqref{eq:Kcal_c}, let
\begin{equation}\label{eq:GK}
\mathcal{G}_K  =  \{    g\in\mathcal{H}^{\gamma/2}:\   g(0)\ge 0, \  \Delta_{\xi} g - ( \Delta_{\xi}  K \ast L)(0) g  - \dd \,( \Delta_{\xi}  K \ast L) \ast g\ge 0 \textrm{ for all $\xi\ge 0$}\}.
\end{equation}
The space $\Gcal_K$ is stochastically invariant with respect to the adjusted version of the forward variance defined in Section \ref{sec_adjusted_forward_process}, and it plays a crucial role in our arguments.

Throughout our study we will make use of the following assumption.
\begin{assumption}\label{assump:kernelsv_0_original}
The kernel $K$ and the function $v_0$ satisfy the following:
\begin{enumerate}
\item\label{assump:kernelsv_0_original:1} $K\in\Kcal$ and $\Delta_{\xi} K$ satisfies \ref{def:Kcal:2} in Definition~\ref{def:Kcal} for all $\xi\ge 0$.
\item\label{assump:kernelsv_0_original:2} $v_0\in\Gcal_K$.
\end{enumerate}
\end{assumption}

The existence and uniqueness in law for the stochastic Volterra equation of the variance process in \eqref{defXV} is guaranteed by the following proposition.

\begin{proposition}\label{prop:existuniq}
Suppose that Assumption \ref{assump:kernelsv_0_original} holds. Then the stochastic Volterra equation for the variance process $V$ in \eqref{defXV} has a unique $\R_+$--valued weak solution. Furthermore, the trajectories of $V$ belong to $\mathcal{H}^{\gamma/2}$ and given $p\ge 1$,
\begin{equation}\label{eq:boundmoments}
\sup_{t\in[0,T]} \E[|V_t|^p]\leq c,\quad T>0,
\end{equation}
where $c<\infty$ is a constant that only depends on $p,T,\lambda,\eta, \gamma, c_K$, and $\|v_0\|_{\Ccal[0,T]}$.
\end{proposition}

\begin{proof}
This result follows from \cite[Theorems 2.1 and 2.3]{abi2019markovian}, with the exception of the last assertion on the bound \eqref{eq:boundmoments}. Following the argument in the proof of \cite[Lemma 3.1]{abi2019affine}, this bound can be shown to depend on $p,T,\lambda,\eta$, $\|v_0\|_{\Ccal[0,T]}$, and $L^2$-continuously on $K|_{[0,T]}$. Note that, thanks to the Fr\'echet--Kolmogorov theorem, the set of restrictions $K|_{[0,T]}$ of nonincreasing kernels satisfying the property \eqref{eq:Kcal_c} for a given $c_K$ and $\gamma$ is relatively compact in $L^2(0,T)$. Maximizing the bounds over all such $K$ yields a bound $c<\infty$ that only depends on $p,T,\lambda,\eta, \gamma, c_K$ and $\|v_0\|_{\Ccal[0,T]}$.
\end{proof}
 
The theoretical results of this study are stated for general kernels $K$ and functions $v_0$ satisfying Assumption \ref{assump:kernelsv_0_original}. This is convenient in order to keep the notation simple. It is also in tune with forward-type stochastic volatility models, such as the rough Bergomi model \cite{bayer2016pricing}. Indeed, thanks to \eqref{eq:boundmoments}, taking expectations in the equation for the variance process in \eqref{defXV} yields the following relation between the function $v_0$ and the initial forward-variance curve $(\E[V_t])$:
\[
v_0(t)=\E[V_t]+\lambda\int_0^t K(t-s)\E[V_s] \dd s.
\]

For the numerical illustrations in Section \ref{sec:numeric} we will use the setting of the rough Heston model \cite{el2019characteristic}, which we summarize in the following example.

\begin{example}\label{ex:roughHeston1}
In the rough Heston model, the kernel $K$ is a fractional kernel 
\begin{equation}\label{frackernel}
K(t)=\frac{t^{\alpha-1}}{\Gamma(\alpha)},
\end{equation}
with $\alpha\in \left(\frac12,1\right]$ and $\Gamma(\alpha)=\int_0^{\infty}x^{\alpha-1}{\rm e}^{-x}\dd x$, and the function $v_0$ is of the form
\begin{equation}\label{eq:v0}
v_0(t)=V_0+\lambda\overline \nu \int_0^t K(s)\dd s,
\end{equation}
where $V_0\geq 0$ is an initial variance and $\overline \nu\geq 0$ is a long term mean reversion level. Assumption \ref{assump:kernelsv_0_original}, with $\gamma=2\alpha-1$, holds in this framework thanks to \cite[Examples 2.3 and 6.2]{abi2019affine} and \cite[Example 2.2]{abi2019markovian}.
\end{example}

Assume that Assumption \ref{assump:kernelsv_0_original} holds. Let $\P$ be the probability measure, and let $\F=(\mathcal F_t)$ be the filtration of the stochastic basis associated to the weak solution $(X,V)$ to \eqref{defXV}. Suppose that $f\in\Ccal_b(\R)$. Our goal is to determine the value process $(P_t)_{0\leq t\leq T}$ of the American option with payoff process $\left( f(X_t)\right)_{0\leq t\leq T}$. We know that $P$ is given by
\begin{equation}
\label{AO}
P_t= \esssup_{\tau\in\mathcal T_{t,T}}\E\left[{\rm e}^{-r(\tau-t)}f(X_{\tau})\vert\mathcal F_t\right],\quad 0\leq t\leq T,
\end{equation}
where $\E$ is the expectation with respect to $\P$ and $\mathcal T_{t,T}$ denotes the set of $\mathbb F$-stopping times taking values in $[t,T]$. In order to compute American option prices, the financial model has to be approximated by more tractable models. In this work, we will consider approximations of the Volterra Heston model resulting from $\Lcal^2$-approximations of the kernel.  In the next section, we describe the approximation procedure.

\subsection{Approximation of the kernel and the Volterra Heston model}
We consider a sequence of kernels $(K^n)_{n\ge 1}$ in $\Lcal^2_{loc}$ and functions $(v_0^n)_{n\ge 1}$ in $\Ccal$. We make the following assumption.

\begin{assumption}\label{assump:kernelsv_0}
The kernels $(K^n)_{n\ge 1}$ and the functions $(v_0^n)_{n\ge 1}$ satisfy the following:
\begin{enumerate}
\item\label{assump:kernelsv_0:1} There exist a constant $\gamma\in (0,2]$ and a locally bounded function $c_{K}:\R_+\to\R_+$ such that $K^n$ satisfies \ref{def:Kcal:1} in Definition~\ref{def:Kcal} for all $n\ge 1$.
\item\label{assump:kernelsv_0:2} $\Delta_{\xi} K^n$ satisfies \ref{def:Kcal:2} in Definition~\ref{def:Kcal} for all $\xi\ge 0$ and $n\ge 1$.
\item\label{assump:kernelsv_0:3} $K^n$ converges to $K$ in $\Lcal^2_{loc}$.
\item\label{assump:kernelsv_0:4} $v^n_0\in\Gcal_{K^n}$, with the constant $\gamma$ of \ref{assump:kernelsv_0:1} for all $n\ge 1$ and $v_0^n$ converges to $v_0$ in $\Ccal$.
\end{enumerate}
\end{assumption}

According to Proposition \ref{prop:existuniq}, under Assumption \ref{assump:kernelsv_0}, for each $n\ge 1$ there exists a unique weak solution $(X^n,V^n)$ to
\begin{equation}\label{eq:approxSVE}
\begin{aligned}
X^n_t&=X_0+\int_0^t \left(r-\frac{V^n_s}{2}\right) \dd s+\int_0^t\sqrt{V^n_s}\left( \rho \dd W^n_s+\sqrt{1-\rho^2}\dd W^{n,\perp}_s\right),\\
V^n_t&=v^n_0(t)-\lambda \int_0^t K^n(t-s)V^n_s\dd s+\eta\int_0^t K^n(t-s)\sqrt{V^n_s} \dd W^n_s,
\end{aligned}
\end{equation}
where $(W^n,W^{n,\perp})$ is a Brownian motion in the corresponding stochastic basis. Furthermore, given $p\ge 1$,
\begin{equation}\label{eq:unifboundmoments}
\sup_{n\ge 1}\sup_{t\in [0,T]} \E^n[|V^n_t|^p]\leq c,\quad T>0,
\end{equation}
with a constant $c<\infty$ which can be chosen to depend only on $p,T,\lambda,\eta, \gamma, c_K$, and $\sup_{n\ge 1}\|v_0^n\|_{\Ccal[0,T]}$, and where $\E^n$ denotes the expectation in the respective probability space. Moreover, the argument in the proof of \cite[Theorem 3.6]{abi2019multifactor} shows that 
\begin{equation}\label{eq:weakconv}
(X^n,V^n)\text{ converges in law to } (X,Y),\text{ in $\Ccal(\R_+,\R^2)$,}\quad\text{as $n\to\infty$}.
\end{equation}
This is a consequence of a more general result proven in Proposition \ref{prposi_converg_un}. 

For completely monotone kernels\footnote{$K$ is completely monotone if $(-1)^m\frac{\dd^m}{\dd t^m}K(t)\ge 0$ for all nonnegative integers $m$.}, an approximation with a sum of exponentials is natural. We briefly explain this procedure  below.

\subsubsection{Approximation with a sum of exponentials}
Assume that the kernel $K$ is completely monotone. By Bernstein's theorem this is equivalent to the existence of a nonnegative Borel measure $\mu$ on $\R_+$ such that
\begin{equation}\label{eq:bernstein}
K(t)=\int_{\R_+}{\rm e}^{-x t}\mu(\dd x).
\end{equation}
As in \cite{abi2019multifactor} and \cite{carmona2000approximation, harms2019affine}, an approximation of the measure $\mu$ in \eqref{eq:bernstein} with a weighted sum of Dirac measures
\begin{equation}\label{eq:sumdiract}
\mu^n=\sum_{i=1}^n c^n_i \delta_{x_i^n}
\end{equation}
yields a candidate approximation of the kernel
\begin{equation}\label{eq:kernelsumexp}
K^n(t)=\int_{\R_+}{\rm e}^{-x t}\mu^n(\dd x)=\sum_{i=1}^n c^n_i {\rm e}^{-x^n_i t}.
\end{equation}
The kernels $(K^n)_{n\ge 1}$ are completely monotone. If, in addition, they are not identically zero, as explained in \cite[Example 6.2]{abi2019affine}, condition \ref{assump:kernelsv_0:2} in Assumption \ref{assump:kernelsv_0} holds.

The representation \eqref{eq:kernelsumexp} yields the following factor-representation for the Volterra equation \eqref{eq:approxSVE} satisfied by the variance process $V^n$: 
\begin{equation}\label{eq:factors}
\begin{aligned}
V^n_t&=v^n_0(t)+\sum_{i=1}^n c_i^n Y^{n,i}_t,\\
Y^{n,i}_t&=\int_0^t (-x_i^nY^{n,i}_s-\lambda V_s^n)\dd s+\int_0^t \eta\sqrt{V_s^n}\dd W^n_s,\quad i=1,\ldots,n.
\end{aligned}
\end{equation}
This representation is convenient because the process $(Y^{n,i})_{i=1}^n$ is an $n$-dimensional Markov process with an affine structure. This observation, together with the convergence in \eqref{eq:weakconv}, was exploited in \cite{abi2019multifactor} in order to approximate European option prices in the rough Heston models employing Fourier methods. The affine structure will also play a crucial role in our study.

We now describe a natural way to determine the weights $c_i^n$ and the points $x_i^n$. We truncate the integral in \eqref{eq:kernelsumexp} and introduce a subdivision $(\eta^n_i)_{i=0}^n$ which is a strictly increasing sequence in $[0,\infty)$. We then define $c_i^n$ and $x_i^n$ as the mass and the center of mass of the interval $[\eta_{i-1}^n,\eta_{i}^n)$, i.e.,
\begin{equation}\label{eq:weightscentermass}
\begin{aligned}
c_i^n&=\int_{[\eta_{i-1}^n,\eta_{i}^n)}\mu(\dd x)=\mu([\eta_{i-1}^n,\eta_{i}^n)),\\
c_i^n x_i^n&=\int_{[\eta_{i-1}^n,\eta_{i}^n)}x\mu(\dd x),\quad i=1,\ldots,n.
\end{aligned}
\end{equation}
In Appendix \ref{sec:app:kernelapprox} we provide sufficient conditions on the measure $\mu$ and the partitions $(\eta^n_i)_{i=0}^n$ that imply condition \ref{assump:kernelsv_0:1} in Assumption \ref{assump:kernelsv_0}.

For the numerical illustrations in Section \ref{sec:numeric} we will use a fractional kernel and a geometric partition which we present in the following example.

\begin{example}\label{ex:roughHeston2}
The fractional kernel \eqref{frackernel} is completely monotone, and in this case 
\[
\mu(\dd x)=\frac{x^{-\alpha}}{\Gamma(1-\alpha)\Gamma(\alpha)}\dd x.
\]
Following \cite{abi2019lifting}, we consider the geometric partition $(\eta_i^n)_{i=0}^n$ given by $\eta_i^n=r_n^{i-\frac{n}{2}}$ for $r_n>1$ such that 
\[
r_n\downarrow 1\quad \text {and }\quad n\log r_n\to\infty\quad\text{as $n\to\infty$}.
\] 
In this setting, the vectors $(c_i^n)$ and $(x_i^n)$ in \eqref{eq:weightscentermass} take the form
\begin{equation}\label{eq:coeffsc_ix_i}
c_i^n= \frac{(r_n^{1-\alpha}-1)}{\Gamma(\alpha) \Gamma(2-\alpha)} r_n^{(1-\alpha)(i-1-n/2)},\quad   \quad  x_i^n=\frac{1-\alpha}{2-\alpha}\frac{r_n^{2-\alpha}-1}{r_n^{1-\alpha}-1} r_n^{i-1-n/2}, \quad \quad i=1, \ldots, n.
\end{equation}
Like in Example \ref{ex:roughHeston1}, along with the kernels $(K^n)_{n\ge 1}$, we consider functions $(v_0^n)_{n\ge 1}$ of the form
\begin{equation*}
v_0^n(t)=V_0+\lambda\overline \nu \int_0^t K^n(s)\dd s.
\end{equation*}
Under this framework Assumption \ref{assump:kernelsv_0} holds\footnote{In the case of a uniform partition $\eta_i^n=i\pi_n$, conditions that ensure  \ref{assump:kernelsv_0:1}-\ref{assump:kernelsv_0:3} in Assumption \ref{assump:kernelsv_0} are studied in \cite{abi2019multifactor}.}. Indeed, Remark \ref{rem:appendix} in Appendix \ref{sec:app:kernelapprox} shows that condition \ref{assump:kernelsv_0:1} holds. As explained in  \cite[Example 6.2]{abi2019affine}, condition \ref{assump:kernelsv_0:2} is a consequence of the complete monotonicity of $K^n$, $n\ge 1$. Condition \ref{assump:kernelsv_0:3} is shown in \cite[Lemma A.3]{abi2019lifting}. This convergence and the considerations in Example \ref{ex:roughHeston1} imply condition \ref{assump:kernelsv_0:4}.  
\end{example}

With the setup of Example \ref{ex:roughHeston2}, since $K^n$ is a $\Ccal^1$-kernel and Assumption \ref{assump:kernelsv_0} holds, \cite[Proposition B.3]{abi2019multifactor} implies that, for each $n\ge 1$, there exists a unique \textit{strong solution}  $(X^n,V^n)$ to \eqref{eq:approxSVE}. Since, in addition, the factor process  $(Y^{n,i})_{i=1}^n$ in \eqref{eq:factors} is a diffusion, classical discretization schemes can be used in order to simulate the trajectories of the variance and log price. Relying on this observation, the numerical study in Section \ref{sec:numeric} uses a simulation-based method in order to approximate American option prices in the rough Heston model. The convergence of the approximated prices is a consequence of the main theoretical findings of our study, which we present in the next section.

\subsection{Main convergence result}
We start by approximating the American option value process $P$ in \eqref{AO} with Bermudan option prices. More precisely, given a nonnegative integer $N$, $T\ge 0$, a partition $(t_i)_{i=0}^N$ of $[0,T]$ with mesh $\pi_N$, and $t\in[0,T],$ we denote by $\mathcal T^N_{t,T}$ the set of $\F$-stopping times taking values in $[t,T]\cap\{t_0,...,t_N\}.$ For any $N\ge 0$, the Bermudan value process is then defined by
\begin{equation}
\label{BO}
P^N_t= \esssup_{\tau\in\mathcal T^N_{t,T}}\E\Big[{\rm e}^{-r(\tau-t)}f(X_{\tau})\vert\mathcal F_t\Big],\quad 0\leq t\leq T.
\end{equation}

In addition, given $(X^n,V^n)_{n\ge 1}$ weak solutions to \eqref{eq:approxSVE}, we define the corresponding American option prices
\begin{equation}
\label{AOn}
P^n_t= \esssup_{\tau\in\mathcal T^n_{t,T}}\E^n\left[{\rm e}^{-r(\tau-t)}f(X^n_{\tau})\vert\mathcal F^n_t\right],\quad 0\leq t\leq T
\end{equation}
and Bermudan option prices
\begin{equation}
\label{BO_N}
P^{N,n}_t=\esssup_{\tau\in\mathcal T^{N,n}_{t,T}}\E^n\Big[{\rm e}^{-r(\tau-t)}f(X^n_{\tau})\vert\mathcal F^n_t\Big],\quad\quad 0\leq t\leq T.
\end{equation}
In the previous definitions, $(\Fcal^n_t)$ is the filtration and $\E^n$ is the expectation on the stochastic basis associated to the weak solution to \eqref{eq:approxSVE}. The sets $\mathcal T^n_{t,T}$, $\mathcal T^{N,n}_{t,T}$ are defined similarly to $\mathcal T_{t,T}$, $\mathcal T^{N}_{t,T}$ on this stochastic basis.

Theorem \ref{thm:main} below is our main theoretical result. It implies, in particular, that the approximated American option prices $P_0^n$ converge to the prices $P_0$ in the original Volterra Heston model. 
\begin{theorem}\label{thm:main}
Suppose that Assumptions \ref{assump:kernelsv_0_original} and  \ref{assump:kernelsv_0} hold. Let $(X,V)$ and $(X^n,V^n)$ be the unique weak solutions to \eqref{defXV} and \eqref{eq:approxSVE}, respectively. For a function $f\in\Ccal_b(\R)$ define $P$, $P^N$, $P^{n}$, and $P^{N,n}$ as in \eqref{AO}, \eqref{BO}, \eqref{AOn}, and \eqref{BO_N}, respectively. Then
\begin{equation}
\label{convbermud}
P^{N,n}_{t_i}\textrm{ converges in law to }P^N_{t_i}\textrm{ as }n\to\infty,\quad N\geq 0,\,0\le i\le N,
\end{equation}
and 
\begin{equation}\label{convamtoam}
\lim_{n\to\infty}P_0^n=P_0.
\end{equation}
\end{theorem}

The proof of Theorem \ref{thm:main} is based on the study of the adjusted forward variance process and the associated Fourier--Laplace transform, which constitutes the main topic of the next section.

\section{Conditional Fourier--Laplace transform}	
\label{sec:conditional}

\subsection{Adjusted forward process}\label{sec_adjusted_forward_process}
In this section we study the adjusted forward process. This infinite-dimensional process was studied in \cite{abi2019markovian} to characterize the Markovian structure of the Volterra Heston model \eqref{defXV}. The adjusted forward process is very useful in order to study path-dependent options such as Bermudan and American options because, as we will see in Section \ref{sec-laplace}, it allows us to better understand the conditional laws of the underlying process by means of the conditional Fourier--Laplace transform.

Assume that Assumption \ref{assump:kernelsv_0_original} holds. Let $\P$ be the probability measure, and let $\F=(\mathcal F_t)$ be the filtration of the stochastic basis associated to the weak solution $(X,V)$ to \eqref{defXV}. The adjusted forward process $(v_t)$ of $V$ is
\begin{equation}
\label{forward_process}
v_t(\xi)= \mathbb{E} \left[   V_{t+\xi} +\lambda\int_{0}^{\xi} K(\xi-s) V_{t+s} \dd s    \vert \mathcal{F}_t   \right],\quad \xi\ge 0.\footnote{We called $(v_t)$ the \textit{adjusted} forward process to distinguish it from the classical Musiela parametrization of the forward process $(\E[V_{t+\cdot}\vert\Fcal_t])$.}
\end{equation}
In particular, the variance process is embedded in the adjusted forward process because $v_t(0)=V_t$. Notice that, thanks to \eqref{eq:boundmoments}, the process $ \left( \int_{0}^{r} K(t+\xi-s) \sqrt{V_s} \dd W_s \right)_{ 0 \le r \le t+\xi }$ is a martingale,  and we can rewrite the adjusted forward process as
\begin{equation}
\label{eq:forward_process}
	v_t(\xi)=v_0(t+\xi) + \int_{0}^{t} K(t+\xi-s) \left[-\lambda V_{s} \dd s  +\eta \sqrt{V_s} \dd W_s\right],\quad \xi\ge 0.
\end{equation}
Moreover, as shown in \cite[Theorem 3.1]{abi2019markovian}, $v_t\in\Gcal_K$ for all $t\ge 0$, i.e., $\Gcal_K$ is stochastically invariant with respect to $(v_t)$.

Similarly, if Assumption \ref{assump:kernelsv_0} holds, we can define the adjusted forward process for the approximating sequence $(V^n)_{n\ge 1}$ by
\begin{equation}
\label{forward_process_approx}
\begin{aligned}
v^n_t(\xi)&=\mathbb{E}^n \left[   V^n_{t+\xi} +\lambda\int_{0}^{\xi} K^n(\xi-s) V^n_{t+s} \dd s    \vert \mathcal{F}^n_t   \right]\\
&=v^n_0(t+\xi) + \int_{0}^{t} K^n(t+\xi-s) \left[-\lambda V^n_{s} \dd s  +\eta \sqrt{V^n_s} \dd W^n_s\right],\quad \xi\ge 0,
\end{aligned}
\end{equation}
and we have $v^n_t(0)=V^n_t$ and $v^n_t\in\Gcal_{K^n}$, for all $t\ge 0$ and $n\ge 1$.

We start with a lemma regarding the regularity for the approximated adjusted forward processes $v^n$, $n\ge 1$.

\begin{lemma} \label{continuity_u} 
	 Let $T,M \ge 0$ and $p> \max \{ 2,  4/\gamma \}$. Suppose that Assumption \ref{assump:kernelsv_0} holds and for $n\ge 1$ define the processes 
	 \[
	 \tilde{v}^n_t(\xi)=v^n_t(\xi)-v^n_0(t+\xi)
	 \] 
	 with $v^n$ as in \eqref{forward_process_approx}. Then 
	\begin{equation*}
	\mathbb{E}^n [\lvert \tilde{v}_{t}^{n}(\xi') -  \tilde{v}_{s}^{n}(\xi) \lvert^p] \le C (\max (\lvert t-s \lvert, \lvert \xi-\xi' \lvert))^{p\gamma/2},\quad (s,\xi),(t,\xi')\in [0,T]\times [0,M],
	\end{equation*}
	where $C$ is a constant that only depends on $p,T,M,\lambda,\eta, \gamma, c_K$, and $\sup_{n\ge 1}\|v_0^n\|_{\Ccal[0,T]}$. As a consequence $(\tilde{v}^n_t(\xi))_{(t,\xi)\in [0,T]\times [0,M]}$ admits an $\alpha$-H\"{o}lder continuous version for any $\alpha<\frac{\gamma}{2}$. Moreover, for this version and for $ \alpha<\frac{\gamma}{2}-\frac{2}{p}$ we have
	\begin{equation*}
	\mathbb{E}^n \left[   \left( \sup_{(t,\xi') \ne (s,\xi)\in[0,T]\times [0,M]}  \frac{\lvert \tilde{v}_{t}^{n}(\xi') -  \tilde{v}_{s}^{n}(\xi) \lvert}{\lvert  (t-s,\xi'-\xi) \lvert^\alpha}   \right)^p  \right] < c,
	\end{equation*}
	where $c<\infty$ is a constant that only depends on $p,\alpha,T,M,\lambda,\eta, \gamma, c_K$, and $\sup_{n\ge 1}\|v_0^n\|_{\Ccal[0,T]}$.
\end{lemma}

\begin{proof}
Thanks to \eqref{forward_process_approx}, we have for $s\leq t$ and $\xi,\xi'\le M$,
\begin{equation*}
\begin{aligned}
\tilde{v}_{t}^{n}(\xi') -  \tilde{v}_{s}^{n}(\xi) &= \tilde{v}_{t}^{n}(\xi')- \tilde{v}_{s}^{n}(\xi') + \tilde{v}_{s}^{n}(\xi') -  \tilde{v}_{s}^{n}(\xi)\\
&=\int_{0}^{s} (K^n(t+\xi'-u)-K^n(s+\xi'-u)) \dd Z^n_u +\int_s^t K^n(t+\xi'-u)  \dd Z^n_u \\
&\quad+ \int_0^s  (K^n(s+\xi'-u)-K^n(s+\xi-u)) \dd Z^n_u
\end{aligned}
\end{equation*} 
where $Z^n_t=-\lambda \int_0^t V^n_{s} \dd s  +\eta \int_0^t \sqrt{V^n_s} \dd W^n_s$. From this point on, using Assumption \ref{assump:kernelsv_0} and the bound \eqref{eq:unifboundmoments}, the argument is analogous to the proof of \cite[Lemma 2.4]{abi2019affine} and it is based on successive applications of Jensen and Burkholder--Davis--Gundy inequalities, and Kolmogorov's continuity theorem; see \cite[Theorem I.2.1]{RY:99}. 
\end{proof}

\begin{remark}
As an immediate consequence of Lemma \ref{continuity_u}, if Assumption \ref{assump:kernelsv_0} holds, then
\begin{equation}\label{eq:boundsupvn}
\sup_{n\geq 1}\E^n\Big[\sup_{t\in[0,T]}V_t^n\Big]\leq c,
\end{equation}
where $c<\infty$ is a constant that only depends on $T,\lambda,\eta, \gamma, c_K$, and $\sup_{n\ge 1}\|v_0^n\|_{\Ccal[0,T]}$.
\end{remark}

We are now able to establish the convergence of the approximated adjusted forward process in the next proposition.
\begin{proposition}\label{prposi_converg_un} 
Suppose that Assumptions \ref{assump:kernelsv_0_original} and \ref{assump:kernelsv_0} hold. Let $X$ (resp., $X^n$) be as in \eqref{defXV} (resp., \eqref{eq:approxSVE}), and let $v$ (resp., $v^n$) be as in \eqref{forward_process} (resp., \eqref{forward_process_approx}). Then, as $n$ goes to infinity, $(X^n_t,v^n_t(\xi))_{(t,\xi)\in\R_+^2}$ converges in law to $(X_t,v_t(\xi))_{(t,\xi)\in\R_+^2}$ in $\Ccal(\R_+^2,\R^2)$.
\end{proposition}

\begin{proof}
This proof is similar to the proof of \cite[Theorem 3.6 and Proposition 4.2]{abi2019multifactor}. We include a short explanation for completeness. Lemma \ref{continuity_u} and Assumption \ref{assump:kernelsv_0}\ref{assump:kernelsv_0:4} imply tightness for the uniform topology of the triple $(X^n,v^n,Z^n)$, where $Z^n_t=-\lambda \int_0^t V^n_{s} \dd s  +\eta \int_0^t \sqrt{V^n_s} \dd W^n_s$. Suppose that $(X,v,Z)$ is a limit point. Thanks to \eqref{forward_process_approx} and \cite[Lemma 3.2]{jaber2019weak}, we have
\begin{align*}
	\nonumber 1 \ast v^n(\xi) &=1 \ast v_0^n(\xi+\cdot)  +   1\ast  \left(\Delta_{\xi} K^n \ast dZ^n   \right)\\
	\nonumber  & = 1 \ast v_0^n(\xi+\cdot) +   \Delta_{\xi} K^n \ast Z^n   \\
	& = 1 \ast v_0^n(\xi+\cdot)  +  \Delta_{\xi} K \ast Z^n  + \left(  \Delta_{\xi} K - \Delta_{\xi} K^n \right)\ast Z^n,\quad \xi\ge 0. 
\end{align*}
In the previous identities, we have used the notation $\Delta_{\xi} K*dZ$ for the stochastic integral $(\Delta_{\xi}K*dZ)_t=\int_0^t K(t-s+\xi)dZ_s$. Assumption \ref{assump:kernelsv_0} and the convergence in law of $(v^n,Z^n)$ towards $(v,Z)$ yield
\begin{equation*} 
	1 \ast v(\xi) =1 \ast v_0(\xi+\cdot)  +   \Delta_{\xi} K \ast Z  \quad \xi\ge 0.
\end{equation*}
One can show, as in \cite[Theorem 3.6]{abi2019multifactor}, that $Z$ is of the form $Z_t=-\lambda \int_0^t V_{s} \dd s  +\eta \int_0^t \sqrt{V_s} \dd W_s$ for some Brownian motion $W$, where $V_t=v_t(0)$, $t\ge 0$. Once again, \cite[Lemma 3.2]{jaber2019weak} implies that
\begin{equation*}
v_t(\xi)=v_0(\xi+t)  + \left(  \Delta_{\xi} K \ast dZ  \right)_t,\quad t,\xi\ge 0.
\end{equation*}
Hence, $V_t=v_t(0)$, $t\ge 0$, is the (unique) weak solution to the stochastic Volterra equation in \eqref{eq:approxSVE} and $v$ is the associated adjusted forward process. Furthermore, since $X$ is completely determined by $V$, one can prove that $(X,V)$ is the unique weak solution to \eqref{defXV} (see also \cite[Theorem 3.5]{abi2019multifactor}). \end{proof}

\subsection{Conditional Fourier--Laplace transforms}
\label{sec-laplace}
This section studies the conditional Fourier--Laplace transform of the log price and the adjusted forward variance in the Volterra Heston model based on previous considerations in \cite{keller2018affine, abi2019markovian, cuchiero2020generalized}. The results of this section will be useful to establish the convergence of Bermudan option prices in the approximated models to the Bermudan option prices in the original model, i.e., \eqref{convbermud} in Theorem \ref{thm:main}, using a dynamic programming approach.

We start by introducing some notation. For a kernel $K\in\Kcal$ define
\begin{equation}\label{eq:dualspace}
\mathcal{G}^*_K   =  \left\{h\in \Bcal_c(\mathbb{R}_+,\mathbb{C}): t\mapsto -{\rm Re}\left(\int_0^{\infty}h(\xi)K(t+\xi) \dd \xi\right)\in\mathcal{G}_K\right\}
\end{equation}
with $\Gcal_K$ as in \eqref{eq:GK}. This space is a \textit{dual space} that we will consider in the computation of the Fourier--Laplace transform of the adjusted forward process. 

The next proposition characterizes the conditional Fourier Laplace transform of the log price $X$ and the adjusted forward variance $v$ through solutions of some Riccati equations. 

\begin{proposition}\label{prop:characteristic_function_X_u_Xn_un}
Suppose that Assumption \ref{assump:kernelsv_0_original} holds, let $X$ be the log price process given by \eqref{defXV}, and let $v$ be the adjusted forward process given by \eqref{forward_process}. Fix $T\ge 0$, $w \in \mathbb{C}$ with ${\rm Re}(w) \in [0,1]$ and $h\in\mathcal{G}^*_K$. Then the conditional Fourier--Laplace transform of $(X,v)$ 
\begin{equation}\label{Fourier_Laplace_transform_ut_xt}
L_t(w,h;X_T,v_T)=\mathbb{E}  \left[    \exp  \left( w X_T + \int_{0}^{\infty}    h(\xi) v_T(\xi) \dd \xi \right)  \vert \mathcal{F}_t  \right],\quad t\leq T,
\end{equation}
can be computed thanks to the following formula:
\begin{equation}\label{Formula_Fourier_Laplace_transform_ut_xt}
L_t(w,h;X_T, v_T)= \exp \left( w (X_t+r(T-t)) + \int_{0}^{\infty} \Psi (T-t,\xi;w,h)v_t(\xi) \dd \xi \right),
\end{equation}
where $\Psi$ satisfies
\begin{equation}\label{eq:invRicatti}
\xi\mapsto \Psi (t,\xi;w,h)\in\Gcal_K^*,\quad t\ge 0,
\end{equation} 
and it is a solution to the following Riccati equation:
\begin{equation}
\label{ricattieq}
\Psi (t,\xi;w,h) =  h(\xi-t) \mathbf{1}_{ \{\xi \ge t \}} + \mathcal{R} \left(w, \int_{0}^{\infty}\Psi(t-\xi,z;w,h) K(z) \dd z \right)\mathbf{1}_{ \{\xi < t \}},\quad t,\xi\ge 0,
\end{equation}
and the operator $\mathcal R$ is defined by 
\begin{align}
\label{equ_r_phi_r_psi_u_x}  
\mathcal{R}(w, \varphi)= \frac{1}{2} (w^2 -w) + \left(  \rho \eta w - \lambda + \frac{\eta^2}{2}      \varphi  \right) \varphi.
\end{align}
Moreover, if 
\[
{\rm Re}(w)=0,\quad\int_{0}^{\infty}    {\rm Re}(h(\xi)) v_T(\xi) \dd \xi\le 0,
\]
then
\[
 \int_{0}^{\infty} {\rm Re}(\Psi (T-t,\xi;w,h))v_t(\xi) \dd \xi \le 0,\quad t\le T.
\]
\end{proposition}

\begin{remark}\label{rem:Psivspsi}
Existence of solutions to \eqref{ricattieq} satisfying \eqref{eq:invRicatti} is shown in Appendix \ref{sec:app:ricatti} (see Propostion \ref{prop_existence_Psi}). Notice that by setting
\begin{equation}\label{eq:defpsi}
\psi(t)= \int_{0}^{\infty}\Psi(t,\xi;w,h) K(\xi)\dd \xi,
\end{equation}
then the Riccati equation \eqref{ricattieq} can be recast as the following \textit{Riccati--Volterra} equation for $\psi$:
\begin{equation}
\label{eq_psi}\psi(t)=\int_{0}^{\infty} h(\xi) K(t+\xi) \dd \xi  + (K \ast \mathcal{R}\left( w,  \psi(\cdot)  \right))(t),
\end{equation}
and we have the identity
\begin{equation}\label{riccati_alt}
\Psi (t,\xi;w,h)  =  h(\xi-t) \mathbf{1}_{ \{\xi \ge t \}} + \mathcal{R} \left(w, \psi(t-\xi) \right)\mathbf{1}_{ \{\xi < t \}}.
\end{equation}
We also point out the similarity between Proposition \ref{prop:characteristic_function_X_u_Xn_un} and \cite[Proposition 4.6]{gatheral2019affine}. In this regard, we highlight that in our formulation we allow the argument $w$ to be complex and we specify an invariant space $\Gcal^*$ for the Riccati equation \eqref{ricattieq}. These points are important for the proof of our main convergence result.
\end{remark}

\begin{proof}[Proof of Proposition \ref{prop:characteristic_function_X_u_Xn_un}]
Let $\Psi$ be a solution to \eqref{ricattieq}, satisfying \eqref{eq:invRicatti} (see Proposition \ref{prop_existence_Psi}). To simplify notation, throughout this proof we will omit the parameters $w$ and $h$. Let $Z$ be the semimartingale $Z_t=-\lambda \int_0^t V_{s} \dd s  +\eta \int_0^t \sqrt{V_s} \dd W_s$, $\psi$ be as in \eqref{eq:defpsi}, and set
\[
\widetilde{Y}_t=  \int_{0}^{\infty} \Psi (T-t,\xi)(v_t(\xi)-v_0(\xi+t))\dd \xi.
\]
The identity \eqref{eq:forward_process}, equation \eqref{ricattieq}, the stochastic Fubini theorem (see \cite[Theorem 65]{P05}), and a change of variables yield

\begin{equation}\label{eq_Y_t_aux}
\begin{split}
 \widetilde{Y}_t  &=  \int_{0}^{\infty}\int_0^t \Psi (T-t,\xi) K(t+\xi-s)\dd Z_s \dd\xi \\
  &=  \int_0^t\int_{T-t}^{\infty} h(\xi-T+t)K(t+\xi-s) \dd\xi \dd Z_s+\int_0^t\int_{0}^{T-t} \mathcal{R}(\psi(T-t-\xi))K(t+\xi-s) \dd\xi \dd Z_s    \\
  &=\int_0^t\int_{T-s}^{\infty} h(\xi-T+s)K(\xi) \dd\xi \dd Z_s+\int_0^t\int_{t-s}^{T-s} \mathcal{R}(\psi(T-s-\xi))K(\xi) \dd\xi \dd Z_s.
\end{split}
\end{equation}
Equation  \eqref{eq_psi} implies that 
\begin{equation}\label{int_Psi_Y_aux}
 \psi(T-s)= \int_{T-s}^{\infty} h(\xi-T+s) K(\xi) \dd\xi   + \int_{0}^{T-s} \mathcal{R}(\psi(T-s-\xi))K(\xi)d\xi.
\end{equation}
We then plug \eqref{int_Psi_Y_aux} into \eqref{eq_Y_t_aux}  and obtain	
\begin{equation}\label{int_psi_Y}
\widetilde{Y}_t = \int_0^t\psi(T-s) \dd Z_s-\int_0^t\int_{0}^{t-s} \mathcal{R}(\psi(T-s-\xi))K(\xi) \dd\xi \dd Z_s.	
\end{equation}
We deduce, thanks to \eqref{int_psi_Y} and the stochastic Volterra equation for the variance process,  the following semimartingale dynamics for the process $\widetilde{Y}$:
\begin{equation}\label{eq_final_dY_t}
	\begin{split}
 \dd \widetilde{Y}_t & = \psi(T-t)\dd Z_t -  \left(\mathcal{R}(\psi(T-t))\int_{0}^{t} K(t-s) \dd Z_s\right) \dd t\\
  & =\psi(T-t)\dd Z_t - \mathcal{R}(\psi(T-t))(V_t-v_0(t))\dd t.\\
  \end{split}
	\end{equation}
On the other hand, similar calculations show that
\begin{equation}\label{eq:Fouriercondinitial}
\int_{0}^{\infty} \Psi (T-t,\xi)v_0(\xi+t)\dd \xi=\int_0^{\infty}h(\xi)v_0(\xi+T)\dd\xi+\int_0^{T-t}\mathcal{R}(\psi(\xi))v_0(T-\xi)\dd\xi.
\end{equation}
Define the process $Y$ as
\[
Y_t=\widetilde{Y}_t+\int_{0}^{\infty} \Psi (T-t,\xi)v_0(\xi+t)\dd \xi=\int_{0}^{\infty} \Psi (T-t,\xi)v_t(\xi) \dd \xi.
\]
From \eqref{eq_final_dY_t} and \eqref{eq:Fouriercondinitial} we obtain the following semimartingale dynamics for $Y$:
\begin{equation}\label{eq_final_dY_t_new}
dY_t=\psi(T-t)\dd Z_t  - \mathcal{R}(\psi(T-t))V_t\dd t.
\end{equation}
Consider now the semimartingale
\[
M_t=\exp (w(X_t-rt)+Y_t).
\]
From \eqref{eq_final_dY_t_new} and It\^o's formula, we  obtain
\begin{eqnarray}
\frac{\dd M_t}{M_t} & =  & w \dd X_t - w r\dd t +\dd Y_t +\frac{1}{2} w^2\dd \,\langle X \rangle_t +\frac{1}{2} \dd \, \langle Y \rangle_t +w  \dd\, \langle X,Y \rangle_t\nonumber\\
& = &  -\frac{w}{2} V_t \dd t + w \sqrt{V_t}\dd B_t + \psi(T-t)\dd Z_t  - \mathcal{R}(\psi(T-t))V_t\dd t \nonumber\\
 & &  + \frac{1}{2} w^2 V_t \dd t + \frac{1}{2}  \psi^2(T-t)\eta^2 V_t \dd t +  \rho \eta w \psi(T-t) V_t \dd t, 
\nonumber
\end{eqnarray}
where $B=\rho W+\sqrt{1-\rho^2}W^\perp$. From the definition of $\mathcal R$ in \eqref{equ_r_phi_r_psi_u_x}, we finally get
\begin{equation}\label{eq:martingal_M}
\frac{dM_t}{M_t} =  w \sqrt{V_t}\dd B_t +\psi(T-t)  \eta \sqrt{V_t}\dd W_t.
\end{equation}
$M$ is then a local martingale and 
\[ 
M_T=\exp\left(w(X_T-rT) + \int_{0}^{\infty}h(\xi) v_T(\xi) \dd\xi\right)
\]
since $\Psi(0,\xi)=h(\xi)$. As pointed out in the proof of Proposition \ref{prop_existence_Psi} in Appendix \ref{sec:app:ricatti}, thanks to the continuity of $\int_{0}^{\infty} h(\xi) K(\cdot+\xi) \dd \xi $, the function $\psi$ is a continuous, and hence bounded, function on $[0,T]$ . Using \eqref{eq:martingal_M}, the fact that $\psi$ is bounded, and a similar argument to the one used in \cite[Lemma 7.3]{abi2019affine}, we can show that 
\[
M_t=M_0\exp\left(U_t-\frac12\langle U\rangle_t\right),\quad\text{with } U_t=\int_0^t \left(w \sqrt{V_s}\dd B_s +\psi(T-s)  \eta \sqrt{V_s}\dd W_s\right),
\] is a true martingale. This implies the formula for the Fourier--Laplace transform  \eqref{Formula_Fourier_Laplace_transform_ut_xt}. The last implication in the statement of the proposition is a direct consequence of  \eqref{Formula_Fourier_Laplace_transform_ut_xt}.

\end{proof}

To establish the convergence of approximated Bermudan option prices, we will use convergence results of the conditional Fourier--Laplace transform, which we present in the following section.

\subsection{Convergence of the Fourier--Laplace transform}

Suppose that the kernels $(K^n)_{n\ge 1}$ and the functions $(v_0^n)_{n\ge 1}$ satisfy Assumption \ref{assump:kernelsv_0}. Let $(X^n,V^n)_{n\ge 1}$ be the solutions to \eqref{eq:approxSVE}, and let $(v^n)_{n\ge 1}$ be the corresponding adjusted forward processes as in \eqref{forward_process_approx}. We define, analogously to \eqref{Fourier_Laplace_transform_ut_xt}, the associated conditional Fourier--Laplace transform 
\begin{equation}\label{Fourier_Laplace_transform_ut_xt_approx}
L^n(w,h^n;X^n_T,v^n_T)=\mathbb{E}^n  \left[    \exp  \left( w X^n_T + \int_{0}^{\infty}    h^n(\xi) v^n_T(\xi) \dd \xi \right)  \vert \mathcal{F}^n_t  \right]
\end{equation} 
with $h^n\in\Gcal^*_{K^n}$ and ${\rm Re}(w)\in[0,1]$. Proposition \ref{prop:characteristic_function_X_u_Xn_un} implies that
\begin{equation*}
L^n_t(w,h^n;X^n_T, v^n_T)= \exp \left( w (X^n_t+r(T-t)) + \int_{0}^{\infty} \Psi^n (T-t,\xi;w,h^n)v^n_t(\xi) \dd \xi \right),
\end{equation*}
where $\Psi^n$ solves \eqref{ricattieq} with $h$ replaced by $h^n$ and $K$ replaced by $K^n$. We have the following convergence result for the conditional Fourier--Laplace transforms.

 \begin{proposition}
\label{prop:convFourierLaplace}
Suppose that Assumptions \ref{assump:kernelsv_0_original} and \ref{assump:kernelsv_0} hold.  Let $X$ (resp., $X^n$) be as in \eqref{defXV} (resp., \eqref{eq:approxSVE}), and let $v$ (resp., $v^n$) be as in \eqref{forward_process} (resp., \eqref{forward_process_approx}). Fix  $T\ge 0$, $w \in \mathbb{C}$ with ${\rm Re}(w) \in [0,1]$, and $(h^n)_{n\ge 1}$ with $h^n\in\mathcal{G}^*_{K^n}$, $n\ge 1$. Assume that there is $M\ge 0$ such that
\[
{\rm supp}(h^n)\subseteq [0,M],\,\,n\ge 1;\quad\text{ and }\quad h^n\to h\in\Gcal_K^* \text{ in $\Bcal([0,M],\mathbb C)$,}\quad\text{as $n\to\infty$.}
\]
Then
\[
L^n(w,h^n;X^n_T,v^n_T) \text{ converges in law to } L(w,h;X_T,v_T)\text{ in $\Ccal[0,T]$,}\quad\text{as $n\to\infty$,}
\]
where $L(w,h;X_T,v_T)$ and $L^n(w,h^n;X^n_T,v^n_T)$ are the conditional Fourier--Laplace transforms defined in \eqref{Fourier_Laplace_transform_ut_xt} and \eqref{Fourier_Laplace_transform_ut_xt_approx}, respectively.
 \end{proposition}
 
The proof of Proposition \ref{prop:convFourierLaplace} is based on Proposition \ref{prposi_converg_un} and the following lemma, whose proof can be found in Appendix \ref{sec:app:ricatti}.
 
 \begin{lemma}
 \label{lem:conv_psi}
Assume that the hypotheses of Proposition \ref{prop:convFourierLaplace} hold. Let $\Psi$ (resp., $\Psi^n$) be solutions to the Riccati equation \eqref{ricattieq} with kernel $K$ (resp., $K^n$) and initial condition $h$ (resp., $h^n$). Define
\[
\psi(t)= \int_{0}^{\infty}\Psi(t,\xi;w,h) K(\xi)\dd \xi,\quad \psi^n(t)= \int_{0}^{\infty}\Psi^n(t,\xi;w,h^n) K^n(\xi)\dd \xi.
\]
Then, as $n$ goes to infinity, $\psi^n$ converges to $\psi$ in $\Ccal[0,T]$. Moreover, letting $\widetilde{M}=\max\{M,T\}$, the support of $\Psi^n(t,\cdot;w,h^n)$ is contained in $[0,\widetilde{M}]$ for  all $n\ge 1$ and $t\le T$, and $\Psi^n(t,\cdot;w,h^n)$ converges to $\Psi(t,\cdot;w,h)$ in $\Bcal([0,\widetilde{M}],\mathbb C)$ uniformly in $t\in [0,T]$.
 \end{lemma}

 \begin{proof}[Proof of Proposition \ref{prop:convFourierLaplace}]
By Proposition \ref{prposi_converg_un} and Skorohod's representation theorem we can construct $(X^n,v^n)$ and $(X,v)$ on the same probability space such that, as $n$ goes to infinity, $(X^n,v^n)$ converges almost surely to $(X,v)$ in $\Ccal(\R_+^2,\R^2)$. This observation and Lemma \ref{lem:conv_psi} imply that
\[
\lim_{n\to\infty}\sup_{t\in[0,T]}\left|w X^n_t+\int_0^{\infty}\Psi^n(t,\xi;w,h^n)v^n_t(\xi)\dd \xi-w X_t-\int_0^{\infty}\Psi(t,\xi;w,h)v_t(\xi)\dd\xi\right|=0,\quad\text{a.s.}
\]
Hence, $w X^n+\int_0^{\infty}\Psi^n(\cdot,\xi;w,h^n)v^n_t(\xi)\dd \xi$ converges in law to $w X+\int_0^{\infty}\Psi(\cdot,\xi;w,h)v_t(\xi)\dd\xi$ in $\Ccal[0,T]$. An application of the continuous mapping theorem with the exponential function, together with Proposition \ref{prop:characteristic_function_X_u_Xn_un}, yields the conclusion.
 \end{proof}
 
We now possess all the elements necessary for the proof of Theorem \ref{thm:main}. 

\section{Proof of the main convergence result}
\label{sec:proof}

We break down the argument into different parts. We start by establishing, in the next section, the convergence of the Bermudan option prices as stated in \eqref{convbermud}. To this end, we will consider a more general payoff structure that is better suited for an inductive argument.

\subsection{Convergence of Bermudan option prices}
Throughout this section we will use the notation
\[
\langle h_1, h_2\rangle=\int_0^{\infty}h_1(\xi) h_2(\xi)\dd \xi
\]
for $h_1\in \Bcal_c(\mathbb R_+,\mathbb C)$ and $h_2 \in \Ccal$. In addition, for a given finite set of indices $J$, we define
\begin{equation*}
\Dcal_J=\{(x,(\eta_j)_{j\in J})\in(\R,\C^{\# J} ):\ {\rm Re}(\eta_j)\le 0\text{ for all $j\in J$}\}.
\end{equation*}

We will consider options with intrinsic payoff processes $(Z_t)_{0\leq t\leq T}$ defined as 
\begin{equation}\label{eq:payoffZ}
Z_t=\left\lbrace\begin{array}{ll}
f(X_t), & \textrm{ for }0\leq t<T,\\
g(X_T, (\langle h_j,v_T\rangle)_{j\in J}), &  \textrm{ for } t=T,
\end{array}\right.
\end{equation}
where $v$ denotes the adjusted forward process \eqref{forward_process}, $J$ is a finite set of indexes, $f\in\mathcal C_b(\R)$, $g\in\mathcal C_b(\mathcal D_J)$, $h_j\in\mathcal G_K^*$ for all $j\in J$, and $(X_T,(\langle h_j,v_T\rangle)_{j\in J})\in\mathcal D_J$. In this setting, the Bermudan option discrete value process  over the grid $(t_i)_{i=0}^N$ takes the form 
\begin{equation}
\label{BO2}
U^N_i= \esssup_{\tau\in\mathcal T^N_{t_i,T}}\E\Big[{\rm e}^{-r(\tau-t_i)}Z_{\tau}\vert\mathcal F_{t_i}\Big],\quad 0\leq i\leq N.
\end{equation}

For the approximating models, and in an analogous manner, we will consider options with payoff processes $(Z^n_t)_{0\leq t\leq T}$ defined as
\begin{equation}\label{eq:payoffZn}
Z^n_t=\left\lbrace\begin{array}{ll}
f(X^n_t), & \textrm{ for }0\leq t<T,\\
g(X^n_T, (\langle h^n_j,v^n_T\rangle)_{j\in J}), &  \textrm{ for } t=T,
\end{array}\right.
\end{equation}
where $h_j\in\mathcal G_{K^n}^*$, for all $j\in J$, and $(X^n_T,(\langle h^n_j,v^n_T\rangle)_{j\in J})\in\mathcal D_J$. The Bermudan option discrete value process,  in the approximated model and over the grid $(t_i)_{i=0}^N$, takes the form 
\begin{equation}\label{BO2n}
U^{N,n}_i= \esssup_{\tau\in\mathcal T^N_{t_i,T}}\E^n\Big[{\rm e}^{-r(\tau-t_i)}Z^n_{\tau}\vert\mathcal F^n_{t_i}\Big],\quad 0\leq i\leq N.
\end{equation}
The following is the main result of this section.
\begin{theorem}\label{T:convBermuda}
Suppose that Assumptions \ref{assump:kernelsv_0_original} and  \ref{assump:kernelsv_0} hold. Let $X$ (resp., $X^n$) be as in \eqref{defXV} (resp., \eqref{eq:approxSVE}), and let $v$ (resp., $v^n$) be as in \eqref{forward_process} (resp., \eqref{forward_process_approx}). Fix $T\ge 0$, $J$ a finite set of indexes, $f\in\mathcal C_b(\R)$, $g\in\mathcal C_b(\mathcal D_J)$, and $(h^n)_{n\ge 1}$ with $h^n\in\mathcal{G}^*_{K^n}$, $n\ge 1$. Assume that there is $M\ge 0$ such that
\[
{\rm supp}(h^n)\subseteq [0,M],\,\,n\ge 1;\quad\text{ and }\quad h^n\to h\in\Gcal_K^* \text{ in $\Bcal([0,M],\mathbb C)$,}\quad\text{as $n\to\infty$.}
\]
Then
\[
U_i^{N,n}\text{ converges in law to }U^N_i,\quad i=0,\ldots,N,\quad\text{as $n\to\infty$,}
\]
where $U^N_i$ and $U_i^{N,n}$ are given by \eqref{BO2} and \eqref{BO2n}, respectively.
\end{theorem}

\begin{proof}
We prove the result by induction on the number of exercise dates $N+1$.\\
\textit{Initialization}. Assume that $N=0$. We just have to prove that
$$\lim_{n\to\infty}g(X_0, (\langle h^n_j,v^n_0\rangle)_{j\in J})=g(X_0, (\langle h_j,v_0\rangle)_{j\in J}).$$
This follows from continuity of $g$ on $\mathcal D_J$, because our hypotheses readily imply 
$$\lim_{n\to\infty}\langle h^n_j,v^n_0\rangle  =  \langle h_j,v_0\rangle.$$
\textit{Induction}.  Assume that the claim holds for Bermudan options with $N$ exercise dates. We have to consider three different cases.
\begin{itemize}
\item[(1)] Suppose that $g$ on $\mathcal D_J$ has the form
\begin{equation*}
g(x,(\eta_j)_{j\in J})={\rm Re}\left(\sum_{k\in I}c_k\exp\left(i\left(\nu_k x+\sum_{j\in J}\beta_{j,k} {\rm Im}(\eta_j)\right) +\sum_{j\in J}\alpha_{j,k}{\rm Re}(\eta_j) \right)\right),
\end{equation*}
with $I$ a finite set of indices, $c_k\in\C$, $\nu_k\in\R$, $\alpha_{j,k}\geq 0$, $\beta_{j,k}\in\R$. In this case the value of the option at maturity (in the original Volterra model) is
\[
Z_T={\rm Re}\left(\sum_{k\in I}c_k\exp(i\nu_k X_T + \langle y_k(0),v _T\rangle)\right),
\]
with 
\[
y_k(0)=\sum_{j\in J}\alpha_{j,k} {\rm Re}(h_j) + i\sum_{j\in J}\beta_{j,k} {\rm Im}(h_j),\quad k\in I.
\]
One can verify that for each $k\in I$, $y_k(0)\in\Gcal_K^*$ thanks to the fact that $\alpha_{j,k}\ge 0$, $j\in J$, and the definition of $\Gcal_K^*$ in \eqref{eq:dualspace} and $\Gcal_K$ in \eqref{eq:GK}. Since the process $U^N$ discounted coincides with the Snell envelope of the discounted payoff process, we have
\begin{eqnarray*}
U^N_{{N-1}} & = & \max\left(Z_{t_{N-1}},\ e^{-r(\Delta t_{N-1})}\E\Big[U^N_T\vert\mathcal F_{t_{N-1}}\Big]\right)\\
& = & \max\left(f(X_{t_{N-1}}),\ e^{-r\Delta t_{N-1}}\E\Big[g(X_T,  (\langle h_j,v_T\rangle)_{j\in J})\vert\mathcal F_{t_{N-1}}\Big]\right),
\end{eqnarray*}
where $\Delta t_{N-1}=t_N-t_{N-1}=T-t_{N-1}$. According to the affine transform formula in Proposition \ref{prop:characteristic_function_X_u_Xn_un}, with $w$ being purely imaginary, the value of the option at time $N-1$ is then
\[
U^N_{{N-1}}=\max\left\{f(X_{t_{N-1}}),{\rm e}^{-r\Delta t_{N-1}}{\rm Re}\left(\sum_{k\in I}c_k {\rm e}^{i\nu_k (X_{t_{N-1}}+r\Delta t_{N-1})+\langle y_k(\Delta t_{N-1}),v_{t_{N-1}}\rangle}\right)\right\},
\]
where $y_k(\Delta t_{N-1})\in \Gcal^*_K$ is a solution at time $\Delta t_{N-1}$ of the associated Riccati equation (with initial condition $y_k(0)$), $k\in I$. Similarly, in the approximated model,  we have 
\[
U^{N,n}_{N-1}=\max\left\{f(X^n_{t_{N-1}}),{\rm e}^{-r\Delta t_{N-1}}{\rm Re}\left(\sum_{k\in I}c_k {\rm e}^{i\nu_k (X^n_{t_{N-1}}+r\Delta t_{N-1})+\langle y^n_k(\Delta t_{N-1}),v^n_{t_{N-1}}\rangle}\right)\right\},
\]
where $y^n_k(\Delta t_{N-1})\in\Gcal_{K^n}^*$ is a solution at time $\Delta t_{N-1}$ of the associated Riccati equation with initial condition 
\[
y^n_k(0)=\sum_{j\in J}\alpha_{j,k} {\rm Re}(h^n_j) + i\sum_{j\in J}\beta_{j,k} {\rm Im}(h^n_j)\in \Gcal_{K^n}^*.
\]
Propositions \ref{prposi_converg_un} and \ref{prop:convFourierLaplace} imply that $U^{N,n}_{N-1}$ converges in law to $U^{N}_{N-1}$. To prove that $U^{N,n}_{i}$ converges in law to $U^{N}_{i}$ for $i=0,\ldots,N-2$, we apply Lemma \ref{lem:conv_psi} together with the induction hypothesis in the case of a Bermudan option with maturity $t_{N-1}$, $N$ exercise dates, and final payoff $\hat g(X_{t_{N-1}},  (\langle \hat h_k,v_{t_{N-1}}\rangle)_{k\in I})$, where, for $k\in I$,  $\hat h_k= y_k(\Delta t)$ and 
\[
\hat{g}(x,(\eta_k)_{k\in I})=\max\left\{f(x),,{\rm e}^{-r\Delta t_{N-1}}{\rm Re}\left(\sum_{k\in I} c_k \exp\left(i\nu_k (x+r\Delta t_{N-1})+\eta_k\right)\right)\right\}. 
\]
Notice that $(X_{t_{N-1}},  (\langle \hat h_k,v_{t_{N-1}}\rangle)_{k\in I})\in\mathcal D_I$ thanks to the last implication in Proposition \ref{prop:characteristic_function_X_u_Xn_un}.

\item[(2)] Assume now that $g$ vanishes outside a compact set $\Gamma\subset\mathcal D_J$. \\
Let $\eps>0$. By tightness of the sequence $(X^n_T,v^n_T)$, its convergence to $(X_T,v_T)$, and the convergence of $h^n_j$ to $h_j$ for all $j\in J$, there exists a compact set $\Gamma'\subset\Dcal_J$ such that $\Gamma\subset \Gamma'$ and
\begin{equation}\label{E:restcompactproba0}
\P\left((X_T,(\langle h_j, v_T \rangle)_{j\in J})\notin \Gamma'\right)<\eps,\quad \P^n\left((X_T^n,(\langle h^n_j, v^n_T \rangle)_{j\in J})\notin \Gamma'\right)<\eps,\quad n\ge 1.
\end{equation}
Furthermore, we can assume that there exists a constant $A>0$ such that
\[
\Gamma'=\left\{(x,(\eta_j)_{j\in J})\in\Dcal_J: \ |x|+\max_{j\in J}(|\eta_j|)\leq A\right\}.
\]
Let $\Acal$ be an algebra of functions defined as follows. We say that a function $\hat g$ on $\Dcal_J$ belongs to $\Acal$ if it is of the form
\[
\hat g(x,(\eta_j)_{j\in J})={\rm Re}\left(\sum_{k\in I}c_k\exp\left(2\pi i \left(\frac{n_k}{2A} x +\sum_{j\in J}\frac{m_{k,j}}{2A}{\rm Im}(\eta_j)\right)+\sum_{j\in J}\alpha_{j,k}{\rm Re}(\eta_j)\right)\right),
\]
with $I$ a finite set of indices, $c_k\in\mathbb{C}$, $\alpha_{j,k}\ge 0$, and $n_k$ and $m_{k,j}$ integers. We also define the following compact subset of $\Dcal_J$:
\[
\widetilde{\Gamma}=\left\{(x,(\eta_j)_{j\in J})\in\Dcal_J:\ |x|+\max_{j\in J}|{\rm Im}(\eta_j)|)\leq A\right\}.
\]
Notice that we have $\Gamma'\subset \widetilde{\Gamma}$, and if we denote by $\Acal|_{\widetilde{\Gamma}}$ the restriction of all the functions in $\Acal$ to $\widetilde{\Gamma}$, then $\Acal|_{\widetilde{\Gamma}}$ is a subset of $\Ccal_0(\widetilde{\Gamma},\R)$---the space of continuous functions that vanish at infinity---that satisfies the hypothesis of the Stone--Weierstrass Theorem. Therefore, there exists $\hat g\in\Acal$ such that
\begin{equation}\label{E:supG-GtildeGammatilde}
\sup_{(x,(\eta_j)_{j\in J})\in\widetilde{\Gamma}}|g(x,(\eta_k)_{k\in I})-\hat g(x,(\eta_k)_{k\in I})|\leq \eps.
\end{equation}

Now observe that for all $(x,(\eta_j)_{j\in J})\in\Dcal_J$, there exists $(x',(\eta_j')_{j\in J})\in \widetilde{\Gamma}$ such that $\hat g(x,(\eta_j)_{j\in J})=\hat g(x',(\eta_j')_{j\in J})$. Hence
\begin{equation}\label{E:supGtildeDcal}
\|\hat g\|_{\infty}\leq \eps +\|g\|_{\infty},
\end{equation}
where $\|\cdot\|_{\infty}$ denotes the sup norm on $\Dcal_J$.

Denote by $\widehat U^N$ (resp., $\widehat U^{N,n}$) the value processes for the Bermudan options corresponding to the payoff process $\widehat Z$ (resp., $\widehat Z^n$) obtained by replacing $g$ by $\widehat g$ in \eqref{eq:payoffZ} (resp., \eqref{eq:payoffZn}). As shown in the previous case, we already know that
\begin{equation}\label{eq:convUhat}
\text{$\widehat U^{N,n}_i$ converges in law to $\widehat U^N_i$ for $i=0,\ldots,N-1$.}
\end{equation}  
Moreover, since the process $U^N$ discounted coincides with the Snell envelope of the discounted payoff process, we have
$$
|U^N_i-\widehat U^N_i| \leq  \E\left[ | U^N_{i+1}-\widehat U^N_{i+1}| \bigg| \mathcal F_{t_i}\right],\quad i=0,\ldots,N-1.$$
By iterating this inequality, we deduce
$$
|U^N_i-\widehat U^N_i| \leq \E\left[|g(X_T,(\langle h_j, v_T \rangle)_{j\in J})-\hat g(X_T,(\langle h_j, v_T \rangle)_{j\in J})|\vert \mathcal F_{t_i}\right],\quad i=0,\ldots,N.$$
Therefore, thanks to the inequalities \eqref{E:restcompactproba0}, \eqref{E:supG-GtildeGammatilde}, and \eqref{E:supGtildeDcal},
\begin{equation}\label{eq:UhatU}
\E\Big[|U^N_i-\widehat U^N_i| \Big]\leq \eps (1+\|\hat g\|_{\infty})\leq \eps(1+\eps+\|g\|_{\infty}),\quad i=0,\ldots,N.
\end{equation}
Similarly, we can prove that
\begin{equation}\label{eq:UnhatUn}
\E^n\Big[|U^{N,n}_i-\widehat U^{N,n}_i|\Big]\leq \eps(1+\eps+\|g\|_{\infty}),\quad i=0,\ldots,N,\,n\ge 1.
\end{equation}
Since $\eps$ is arbitrary we conclude, using \eqref{eq:convUhat}, \eqref{eq:UhatU}, and \eqref{eq:UnhatUn}, that $U^{N,n}_i$ converges in law to $U^N_i$, for $i=0,\ldots, N$.

\item[(3)] Suppose now that $g$ belongs to $\mathcal C_b(\Dcal_J)$.\\
Let $\eps>0$ be arbitrary. As before, tightness of the sequence $(X^n_T,v_T^n)$, its convergence to $(X_T,v_T)$, and the convergence of $h^n_j$ to $h_j$, $j\in J$, imply that there is a compact set $\Gamma\subset\Dcal_J$ such that
\begin{equation*}
\P\left((X_T,(\langle h_j, v_T \rangle)_{j\in J})\notin \Gamma\right)<\eps,\quad \P^n\left((X_T^n,(\langle h^n_j, v^n_T \rangle)_{j\in J})\notin \Gamma\right)<\eps,\quad n\ge 1.
\end{equation*}

Let $\varphi:\Dcal_J\to[0,1]$ be a function of compact support such that $\varphi\equiv 1$ on $\Gamma$.

Denote $\overline U^N$ (resp., $\overline U^{N,n}$) the value processes for the Bermudan options corresponding to the payoff process $\overline Z$ (resp., $\overline Z^n$) obtained by replacing $g$ by $\overline g=\varphi g$ in \eqref{eq:payoffZ} (resp., \eqref{eq:payoffZn}). As shown in the previous case, we already know that
\begin{equation}\label{eq:convUoverline}
\text{$\overline U^{N,n}_i$ converges in law to $\overline U^N_i$ for $i=1,\ldots,N-1$.}
\end{equation}
Additionally, we have
\begin{equation}\label{eq:UoverlineU}
\begin{split}
\E\Big[|U^N_i-\overline{U}^N_i|\Big]&\leq \E\left[|g(X_T,\langle h_j, v_T \rangle)_{j\in J}))-\overline g(X_T,\langle h_j, v_T \rangle)_{j\in J}))|\right]\\
&\leq \eps \|g\|_{\infty},
\end{split}
\end{equation}
and
\begin{equation}\label{eq:UnoverlineUn}
\E^n\Big[|U^{N,n}_i-\overline{U}^{N,n}_i|\Big]\leq  \eps \|g\|_{\infty}.
\end{equation}
Since $\eps$ is arbitrary we conclude, from \eqref{eq:convUoverline}, \eqref{eq:UoverlineU}, and \eqref{eq:UnoverlineUn}, that $U^{N,n}_i$ converges in law to $U^N_i$, for $i=0,\ldots, N$.
\end{itemize}
\end{proof}

\subsection{Approximation of American options with Bermudan options}

The following theorems establish the convergence of Bermudan option prices towards American option prices and they are crucial in order to prove Theorem \ref{thm:main}.

\begin{theorem}
\label{T:convbermudaamerican}
Suppose that Assumption \ref{assump:kernelsv_0_original} holds. Let $(X,V)$ be the unique weak solution to \eqref{defXV}. For a function $f\in\Ccal^2_b(\R)$ consider the American and Bermudan option prices given by \eqref{AO} and \eqref{BO}, respectively. Then
\begin{equation}\label{eq:BvsA}
0\leq P_0-P_0^N\leq c\left(1+\E\Big[\sup_{t\in[0,T]}V_t\Big]\right)\pi_N,
\end{equation}
where $\pi_N$ is the mesh of the partition $(t_i)_{i=0}^N$ and $c$ is a constant that only depends on $r,T$ and $\|f^{(m)}\|_{\Ccal[0,T]}$, $m=0,1,2$.
\end{theorem}

\begin{proof}
We obviously have $0\leq P_0-P^N_0$. Let $\varepsilon>0$. There exists $\tau_\varepsilon^*\in\mathcal T_{0,T}$, $\varepsilon$-optimal in the sense that
$$P_0\leq \E\Big[{\rm e}^{-r\tau_\varepsilon^*}f(X_{\tau_\varepsilon^*})\Big]+\varepsilon.$$
Now, we introduce the lowest stopping time taking values in $\{t_0,...,t_N\}$, greater than $\tau_\varepsilon^*$; that is
$$\tau_\varepsilon^{N,*}= \inf\{ t_k:\ t_k\geq \tau^*_{\varepsilon}\}.$$
We have that $\tau_\varepsilon^{N,*}$ belongs to $\mathcal T^N_{0,T}$.
Since the drift and the quadratic variation of $X$ are affine in $V$, applying It\^o's formula to the process $\left({\rm e}^{-rt}f(X_t)\right)_{0\leq t\leq T}$ between $\tau^{*}_\varepsilon$ and $\tau^{N,*}_\varepsilon$ yields
\begin{eqnarray}
P_0-P^N_0 & \leq & c\E\Big[\int_{\tau^{*}_\varepsilon}^{\tau^{N,*}_\varepsilon}(1+V_s)\, ds\Big]+\varepsilon\nonumber\\
 & \leq &  c\E\Big[(\tau^{N,*}_\varepsilon-\tau^{*}_\varepsilon)\sup_{t\in[0,T]}(1+V_t)\Big]+\varepsilon\nonumber\\
 & \leq &c\left(1+\E\Big[\sup_{t\in[0,T]}V_t\Big]\right)\pi_N +\varepsilon,\nonumber
\end{eqnarray}
where $c$ is a constant that only depends on $r,T$ and $\|f^{(m)}\|_{\Ccal[0,T]}$, $m=0,1,2$. Since $\varepsilon>0$ was arbitrary, we deduce \eqref{eq:BvsA}.
\end{proof}

As a direct consequence of \eqref{eq:boundsupvn} and Theorem \ref{T:convbermudaamerican}, we have the following convergence result for payoffs $f\in\Ccal^2_b(\R)$.
\begin{theorem}
Suppose that Assumptions \ref{assump:kernelsv_0_original} and  \ref{assump:kernelsv_0} hold. Let $(X,V)$ and $(X^n,V^n)$ be the unique weak solutions to \eqref{defXV} and \eqref{eq:approxSVE}, respectively. For a function $f\in\Ccal^2_b(\R)$ define $P$, $P^N$, $P^{n}$ and $P^{N,n}$ as in \eqref{AO}, \eqref{BO}, \eqref{AOn}, and \eqref{BO_N}, respectively.
Then
\begin{equation}\label{convbermam}
\lim_{\pi_N\to 0}\sup_{n\ge 1}|P^{N,n}_0-P^{n}_0|=\lim_{\pi_N\to 0}|P^N_0-P_0|=0.
\end{equation}
\end{theorem}

We are now ready to prove our main theorem.

\begin{proof}[Proof of Theorem \ref{thm:main}]
The convergence in \eqref{convbermud} is a direct consequence of Theorem \ref{T:convBermuda}. For a function $f\in\Ccal^2_b(\R)$, the limit \eqref{convamtoam} follows from \eqref{convbermud} and \eqref{convbermam}. It is then sufficient to show  \eqref{convamtoam} for a function $f\in\Ccal_b(\R)$ knowing that
\begin{equation}\label{eq0}
\lim_{n\to\infty}|P_0^{n,g}-P_0^g|=0,\quad g\in\mathcal C_b^2(\mathbb R),
\end{equation}
where $P_0^{n,g}$ and $P_0^g$ denote the prices at time zero of an American option with payoff $g$ in the approximated model \eqref{eq:approxSVE} and in the Volterra Heston model \eqref{defXV}, respectively.\\
To this end fix $\varepsilon>0$. For all $\tau\in\mathcal T_{0,T}$, $\tau_n\in\mathcal T^n_{0,T}$, and $M> 0$ we have
\begin{equation}\label{eq1}
\begin{aligned}
\mathbb P(|X_{\tau}|>M)&\leq \frac{\mathbb{E}[|X_{\tau}|]}{M}\leq \frac{\mathbb{E}[\sup_{t\in[0,T]}|X_t|]}{M},\\
\mathbb P^n(|X^n_{\tau_n}|>M)&\leq \frac{\mathbb{E}^n[|X^n_{\tau_n}|]}{M}\leq \frac{\mathbb{E}^n[\sup_{t\in[0,T]}|X^n_t|]}{M},\quad n\ge 1.
\end{aligned}
\end{equation}
Moreover, the Burkholder--Davis--Gundy inequality and \eqref{eq:boundsupvn} imply that 
\begin{equation}\label{eq2}
\begin{aligned}
\mathbb E[\sup_{t\in[0,T]}|X_t|]&\leq C(1+\mathbb{E}[\sup_{t\in[0,T]}V_t])\leq c,\\
\mathbb E^n[\sup_{t\in[0,T]}|X^n_t|]&\leq C(1+\mathbb{E}^n[\sup_{t\in[0,T]}V^n_t])\leq c,\quad n\ge 1,
\end{aligned}
\end{equation}
for a constant $c$ that depends only on $T,\lambda,\eta,\gamma,c_K$, and $\sup_{n\ge 1}\|v_0\|_{\mathcal C[0,T]}$ . We conclude, thanks to \eqref{eq1} and \eqref{eq2}, that if $M = \frac{c}{\eps}$, then for all $\tau\in\mathcal T_{0,T}$ and $\tau_n\in\mathcal T^n_{0,T}$
\begin{equation}\label{eq4}
\mathbb P(|X_{\tau}|>M)\leq\varepsilon,\quad \mathbb P^n(|X^n_{\tau_n}|>M) \leq\varepsilon,\quad n\ge 1.
\end{equation}
Let $g\in\mathcal C_{c}^{\infty}(\mathbb R)$ be such that
\begin{equation}\label{eq5}
\sup_{x\in[-M,M]}|f(x)-g(x)|\leq \varepsilon,\quad \|g\|_{\infty}\leq \|f\|_{\infty}.
\end{equation}
For all $\delta>0$ there exists $\tau_{\delta}\in\mathcal T_{0,T}$ such that
\[
\mathbb E[f(X_{\tau_\delta})]>P_0^f-\delta.
\]
The inequalities \eqref{eq4} and \eqref{eq5} imply
\[
\mathbb E[f(X_{\tau_\delta})]\leq P_0^g+\mathbb E[(f-g)(X_{\tau_\delta})]\leq P_0^g +\varepsilon(1+2\|f\|_{\infty}).
\]
Then
\[
P_0^f\leq P_0^g +\varepsilon(1+2\|f\|_{\infty}) +\delta.
\]
Similarly,
\[
P_0^g\leq P_0^f +\varepsilon(1+2\|f\|_{\infty}) +\delta.
\]
Since $\delta>0$ was arbitrary we conclude that
\[
|P_0^f-P_0^g|\leq \varepsilon(1+2\|f\|_{\infty}).
\]
An analogous argument over the approximating models, using the inequalities \eqref{eq4} and \eqref{eq5}, yields
\[
|P_0^{n,f}-P_0^{n,g}|\leq \varepsilon(1+2\|f\|_{\infty}),\quad n\ge 1.
\]
Therefore,
\[
|P_0^{n,f}-P_0^f|\leq |P_0^{n,g}-P_0^g|+2\varepsilon(1+2\|f\|_{\infty}),\quad n\ge 1,
\]
and thanks to \eqref{eq0}
\[
\limsup_{n\to\infty}|P_0^{n,f}-P_0^f|\leq 2\varepsilon(1+2\|f\|_{\infty}).
\]
Since $\varepsilon>0$ was arbitrary, this yields
\[
\lim_{n\to\infty}|P_0^{n,f}-P_0^f|=0.
\]

\end{proof}

\section{Numerical illustrations}
\label{sec:numeric}

In this section we illustrate with numerical examples the convergence and behavior of Bermudan put option prices in the approximated sequence of models. To this end, we consider the framework of the rough Heston model in Example \ref{ex:roughHeston1} and the approximation scheme of Example \ref{ex:roughHeston2}. 

We choose the same model parameters as in \cite{abi2019lifting}, namely
\begin{equation*}
V_0=0.02,\quad \bar{\nu}=0.02,\quad\lambda=0.3,\quad\eta=0.3,\quad\rho=-0.7.
\end{equation*}
We fix a maturity $T=0.5$ and a spot interest rate $r=0.06$. 

In order to compute Bermudan option prices in the approximated model $(X^n,V^n)$ in \eqref{eq:approxSVE}, we apply the Longstaff Schwartz algorithm \cite{longstaff2001valuing} using $10^5$ path simulations. Following the suggestion in \cite{abi2019lifting}, and based on the factor-representation \eqref{eq:factors}, we simulate the trajectories of the variance with a \textit{truncated} explicit-implicit Euler-scheme and the trajectories of the log price with an explicit Euler-scheme. More precisely, given a uniform partition $(s_k)_{k=0}^{N_{time}}$ of $[0,T]$ of norm $\Delta t$, and  $\left(G^k_1\right)_{k\geq 1}$ and $\left(G^k_2\right)_{k\geq 1}$ independent sequences of independent centered and reduced Gaussian variables, we simulate the log price with the scheme
\begin{equation*}
\widehat X_{s_{k+1}}^{n}=  \widehat X_{s_{k}}^{n}  + \left(  r - \frac{\widehat V_{s_k}^{n}}{2} \right) \Delta t + \sqrt{\widehat V_{s_k}^{n\ +}}\sqrt{ \Delta t } \left(\rho  G_1^{k+1} + \sqrt{1- \rho^2}  G_2^{k+1}  \right),\quad \widehat X_{s_{0}}^{n}=X_0,
\end{equation*}
and the variance with the scheme
\begin{equation*}
\widehat V_{s_k}^{n}=v_0^n(s_k)+ \sum_{i=1}^{n} c_i^n\widehat  Y_{s_k}^{n,i},   \quad   \widehat Y_0^{n,i}=0,   \quad i=1, \ldots,n,
\end{equation*}
\begin{equation*}
\widehat Y_{s_{k+1}}^{n,i}= \frac{1}{1 + x_i^{n} \Delta t} \left( \widehat Y_{s_{k}}^{n,i}  - \lambda  \widehat V_{s_k}^{n} \Delta t + \eta \sqrt{\widehat V_{s_k}^{n\ +}}\sqrt{ \Delta t }G_1^{k+1}\right) , \quad i=1, \ldots,n.
\end{equation*}
In this framework the approximation of the initial curve $v_0$ in \eqref{eq:v0} takes the form
\[
v_0^n(s_k) = V_0 + \lambda \bar{\nu} \sum_{i=1}^{n} c_i^{n} \left( \frac{1 - e^{-x_i^n s_k}}{x_i^n}  \right).\footnote{For generality and consistency reasons, we have done the numerical experiments with this approximation of $v_0$. We notice, however, that for the fractional kernel, $v_0$ is given explicitly by $v_0(t)=V_0+\lambda\overline \nu\frac{t^\alpha}{\Gamma(1+\alpha)}.$}
\] 
 We take $N_{time}=500$ and select equidistant exercise times $(t_k)_{i=0}^N$, with $N=50$, within the partition $(s_k)_{k=0}^{N_{time}}$. Given a strike price $K$, for the regressions of the Longstaff Schwartz algorithm we use the linear space of functions generated by functions with argument $S$, corresponding to the log price, and $V$ corresponding to the volatility, of the form 
\[
f_1\left(\frac{S}{K}\right)f_2\left(\frac{V}{\bar{\nu}}\right),\quad f_1,f_2\in\Acal,
\]
where $\Acal$ is given by
\[
\Acal=\{1\}\cup\{{\rm e}^{-z}L_i(z):i=0,1,2\},
\]
and $L_i$ denotes the Laguerre polynomial of order $i$.\footnote{In the framework of our factor-approximation scheme, the prices of the Bermudan options at intermediate times are functions of the price $S$ and the factors $(Y^{n,i})_{i=1}^n$ defined in \eqref{eq:factors}. These functions could be approximated using neural network-based techniques similar to those in \cite{lapeyre2019neural}. Our initial experiments, however, indicate that there is no significant gain in using this more complex approach. This is consistent with similar findings in \cite{bayer2020pricing,goudenege2020machine} for American options prices in the rough Bergomi model.}

To illustrate the convergence of options prices, we fix the parameter $\alpha=0.6$ and choose parameters $r_n>1$ in the kernel approximation such that  
\begin{equation}\label{eq:optrn}
\begin{aligned}
r_n&=\underset{r}{\argmin}\|K-K^r\|^2_{\Lcal^2(0,T)}\\
&=\underset{r}{\argmin} \left(\sum_{i,j\leq n}c_i^rc_j^r\frac{1-{\rm e}^{-(x_i^r+x_j^r)T}}{x_i^r+x_j^r}-2\sum_{i\leq n}c_i^r(x_i^r)^{-\alpha}\gamma(\alpha,Tx_i^r)\right),
\end{aligned}
\end{equation}
where $c_i^r,x_i^r$, $i=1,\ldots,n$, are as in \eqref{eq:coeffsc_ix_i} with $r_n$ replaced by $r$, $K^r$ is the corresponding kernel obtained as a sum of exponentials, and $\gamma(\alpha,x)=\frac{1}{\Gamma(\alpha)}\int_0^x t^{\alpha-1}{\rm e}^{-t}\dd t$ is the lower incomplete gamma function. Table \ref{T:r_n_alpha06} contains the values of the parameter $r_n$ along with the corresponding values of $\|K-K^n\|^2_{\Lcal^2(0,T)}$ for $n=4,10,20,40,200$.\footnote{These values were obtained using the function \textit{fminbnd} in the MATLAB optimization toolbox.} Figure \ref{fig:Pricesin_n} shows Bermudan put option prices for a strike $K=100$, initial prices $S_0=\exp(X_0)$ in $[93,96]$, and $n=4,10,20,40$ number of factors. We also plot the prices obtained for the classical Heston model, which corresponds to the case $n=1$, $c_1^1=1$, and $x_1^1=0$. For each set of prices we indicate the corresponding so-called critical price, this is the greatest value of the initial price for which the Bermudan option price is equal to the payoff. We observe that as $n$ increases the option prices on this interval decrease and as a result the critical price increases. In Figure \ref{fig:boundaryin_nKnK}, we plot the critical-price as a function of the norm $\|K-K^n\|_{\Lcal^2(0,T)}$ for $n=1,4,10,20,40$, where $n=1$ corresponds to the classical Heston model. Computing prices with $n=200$ factors we observe the same critical price as with $n=40$ which illustrates the convergence of the approximated models.  

\begin{table}
\centering
\begin{tabular}{|c|c|c|}
 \hline
 $n$ & $r_n$ & ${\rm norm}^2_n$ \\
 \hline
 \hline
  4 & 50.5458 & 0.3699 \\
  10 & 18.0548 & 0.1125 \\
  20 & 8.8750 & 0.0325 \\
  40 & 4.4737 & 0.0076 \\
  200 & 1.6946 & 1.1166${\rm e}$-04\\                                    
  \hline                          
\end{tabular}
\caption{Values of $r_n$ and ${\rm norm}^2_n=\|K-K^n\|^2_{\Lcal^2(0,T)}$ obtained using \eqref{eq:optrn} with $\alpha=0.6$ and $T=0.5$.}\label{T:r_n_alpha06}
\end{table}

\begin{figure}
\centering
\includegraphics[scale=0.5]{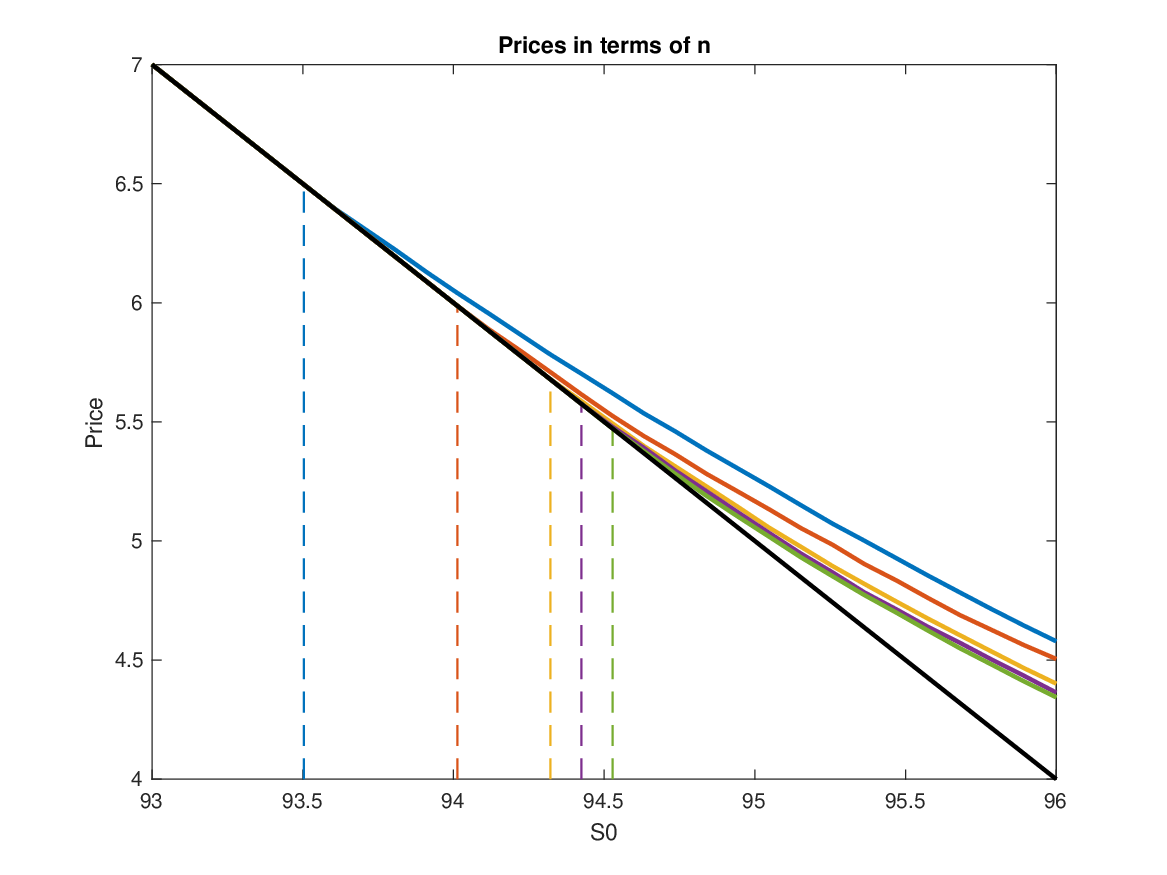}
\caption{Bermudan put option prices in terms of $n$. Payoff (black), Heston model (blue), $n=4$ (red), $n=10$ (yellow), $n=20$ (purple), $n=40$ (green).}\label{fig:Pricesin_n}
\end{figure}

\begin{figure}
\centering
\includegraphics[scale=0.5]{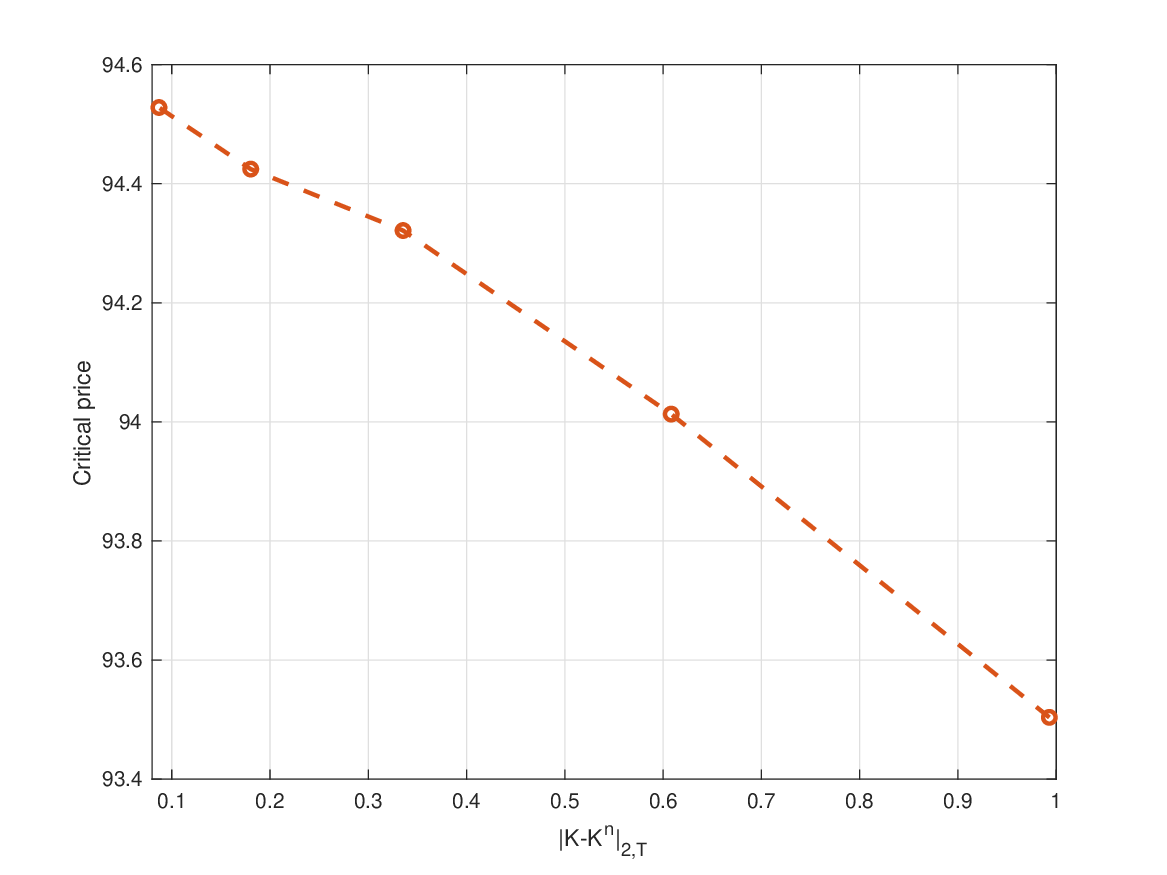}
\caption{Critical prices as a function of $\|K-K^n\|_{\Lcal^2(0,T)}$.}\label{fig:boundaryin_nKnK}
\end{figure}

To study the behavior of Bermudan put option prices with respect to the parameter $\alpha$, and taking into account our previous findings, we proxy the prices in the rough Heston model using the approximated model with $n=40$ factors. We consider the same parameters as in the previous example with the exception of $\alpha$. The parameter $r_{40}$ is chosen as in \eqref{eq:optrn} depending on the parameter $\alpha$ of the fractional kernel $K$. We compute prices and critical prices for $\alpha=0.6,0.7,0.8,0.9,1$. Figure \ref{fig:Pricesinalpha} shows the Bermudan option prices obtained for these values of $\alpha$ and Figure \ref{fig:boundaryin_alpha} displays the critical price as a function of $\alpha$. As $\alpha$ increases, we observe a similar behavior as the one obtained by increasing $\|K-K^n\|_{\Lcal^2(0,T)}$ in our previous example. More precisely, as the regularity of the paths in the model increases, i.e., $\alpha$ increases, the prices of the option increase and the critical price decreases. This is consistent with similar findings reported in \cite{horvath2017functional} within the context of the rough Bergomi model and it could be a consequence of the fact that for smaller values of $\alpha$ the variance has rougher paths and spends more time in a neighborhood of zero.  

\begin{figure}
\centering
\includegraphics[scale=0.5]{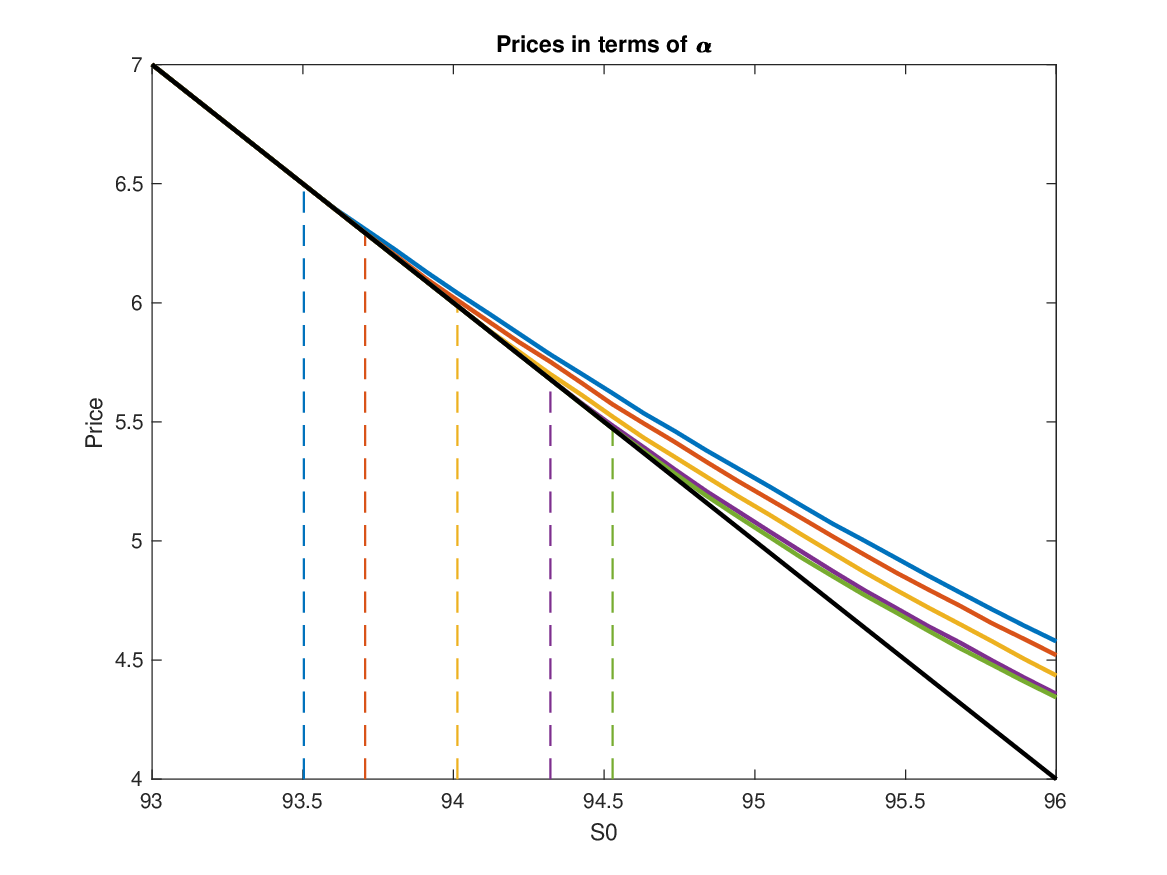}
\caption{Bermudan put option prices and critical prices in terms of $\alpha$. Payoff (black), $\alpha=1$ (blue), $\alpha=0.9$ (red), $\alpha=0.8$ (yellow), $\alpha=0.7$ (purple), $\alpha=0.6$ (green).}\label{fig:Pricesinalpha}
\end{figure}

\begin{figure}
\centering
\includegraphics[scale=0.5]{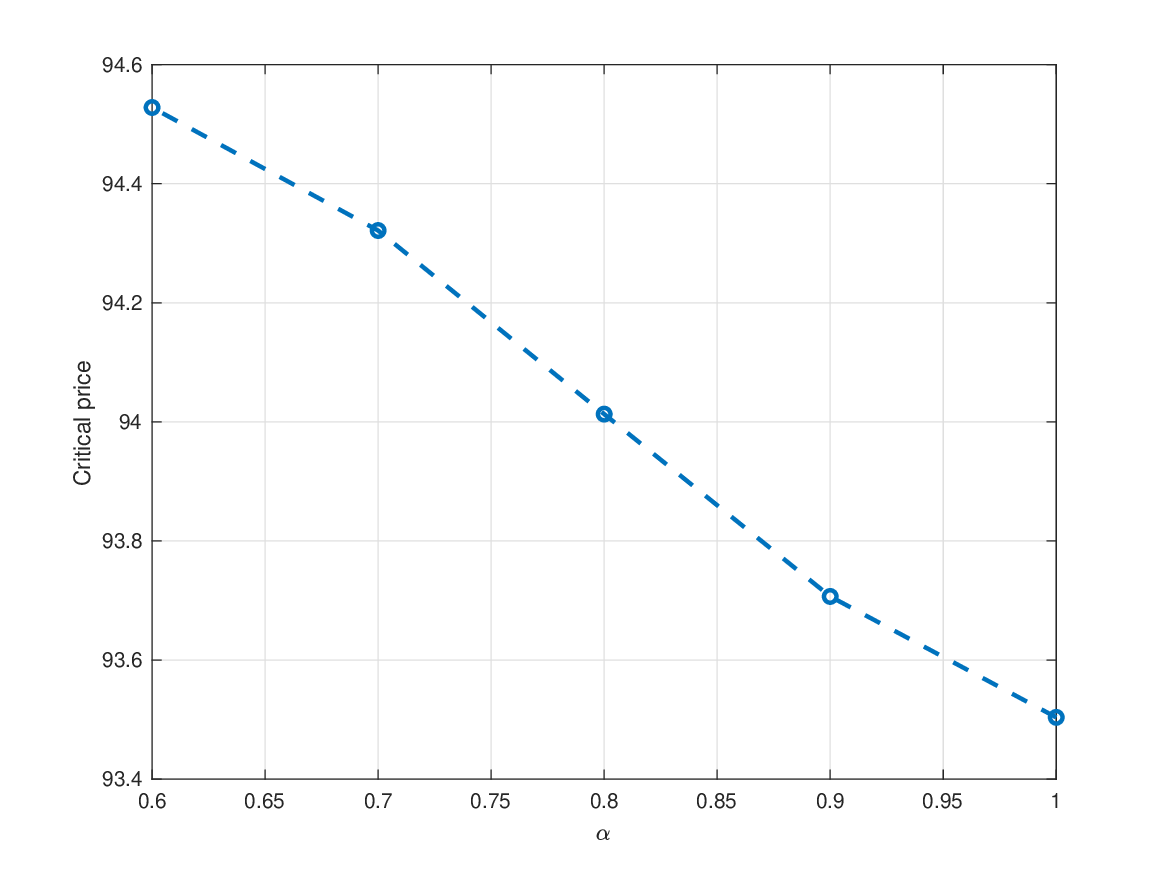}
\caption{Critical prices as a function of $\alpha$.}\label{fig:boundaryin_alpha}
\end{figure}

To illustrate the impact of the initial spot variance, we compare in Figure \ref{fig:boundaryin_V0} the levels of the critical price for different values of $V_0$ in the rough Heston model with $\alpha=0.6$ and the classical Heston model. The critical price seems to depend almost linearly on the initial spot variance $V_0$ in both the classical and the rough Heston model. In the rough Heston model the critical price, and hence the Bermudan option prices, appear to be slightly less sensitive to the initial level of the variance. This could be a result of the difference in sensitivity, with respect to $V_0$, of the time spent around zero by the trajectories in the classical and rough Heston models. 

We also plot in Figure \ref{fig:boundaryin_T} the critical prices for different maturities and for $V_0=0.06$. We observe that for short maturities the sensitivity is higher in the classical Heston model than in the rough Heston model. This is coherent with the previously described behavior with respect to the initial variance level.

\begin{figure}
\centering
\includegraphics[scale=0.5]{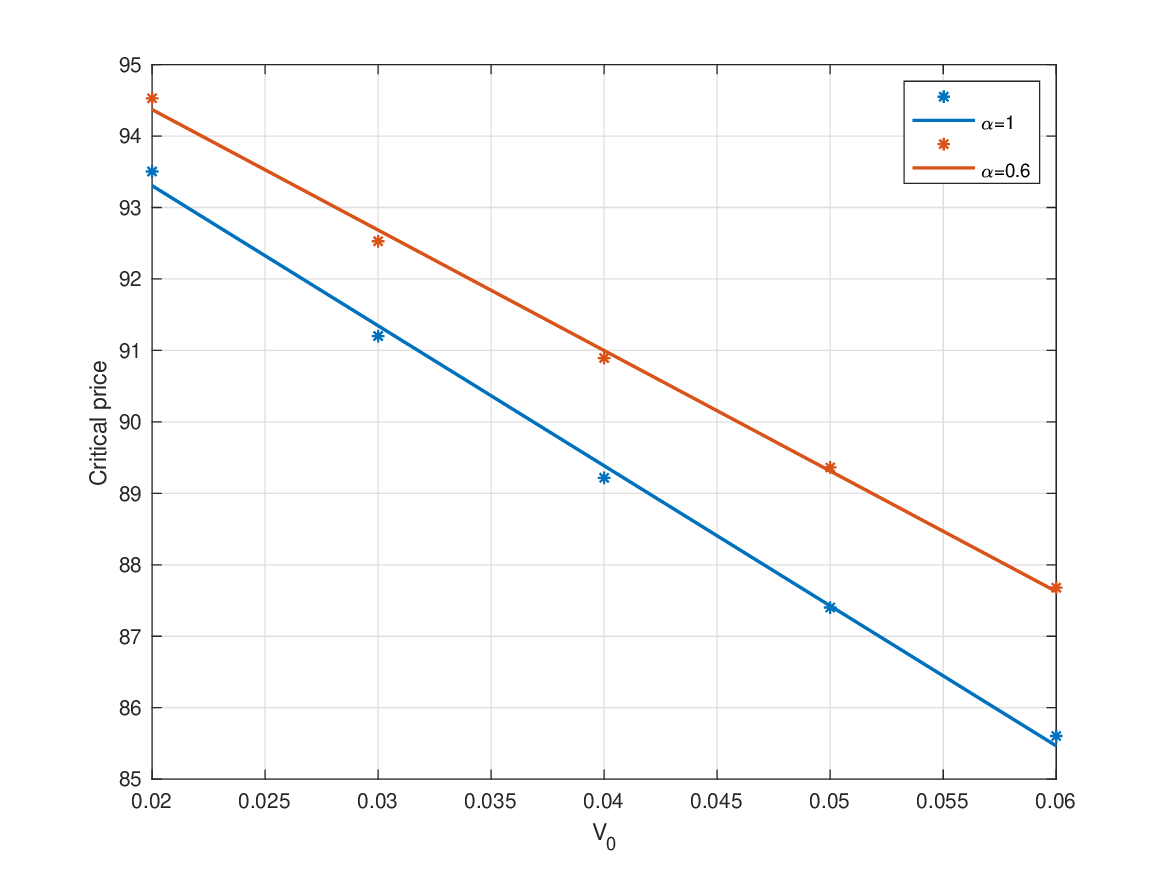}
\caption{Critical prices for $\alpha=0.6,1$ and $V_0=0.02+k*0.01$, $k=0,1,2,3,4$. The solid lines represent the linear regressions.}\label{fig:boundaryin_V0}
\end{figure}

\begin{figure}
\centering
\includegraphics[scale=0.5]{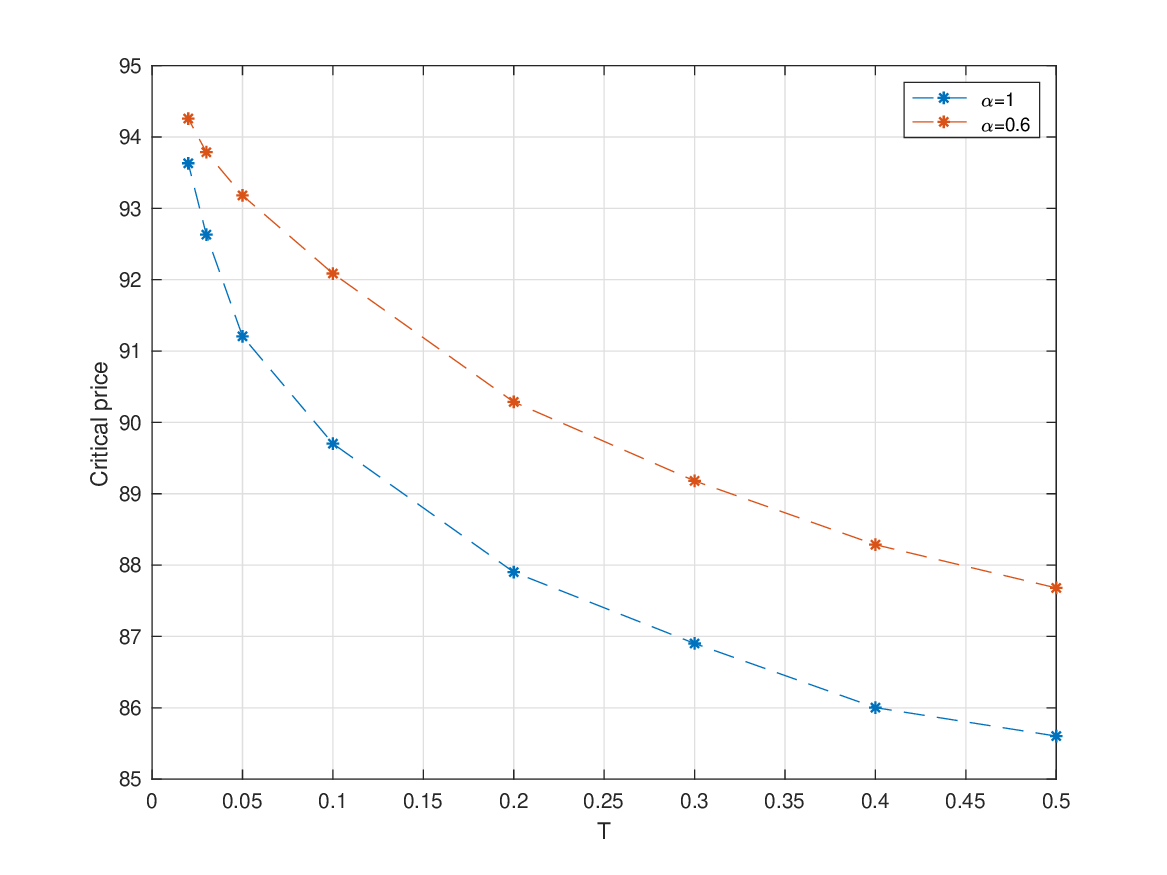}
\caption{Critical prices for $\alpha=0.6,1$, $V_0=0.06$ and different maturities $T$. Specifically, $T=0.02, 0.03, 0.05, 0.1, 0.2, 0.3, 0.4, 0.5$.}\label{fig:boundaryin_T}
\end{figure}

To finish, we numerically illustrate in Figure \ref{fig:Pricesin_N} the convergence of Bermudan put option prices to American put option prices in the rough Heston model with $\alpha=0.6$. Figure \ref{fig:boundaryin_N} shows the convergence of the corresponding critical prices. 
\begin{figure}
\centering
\includegraphics[scale=0.5]{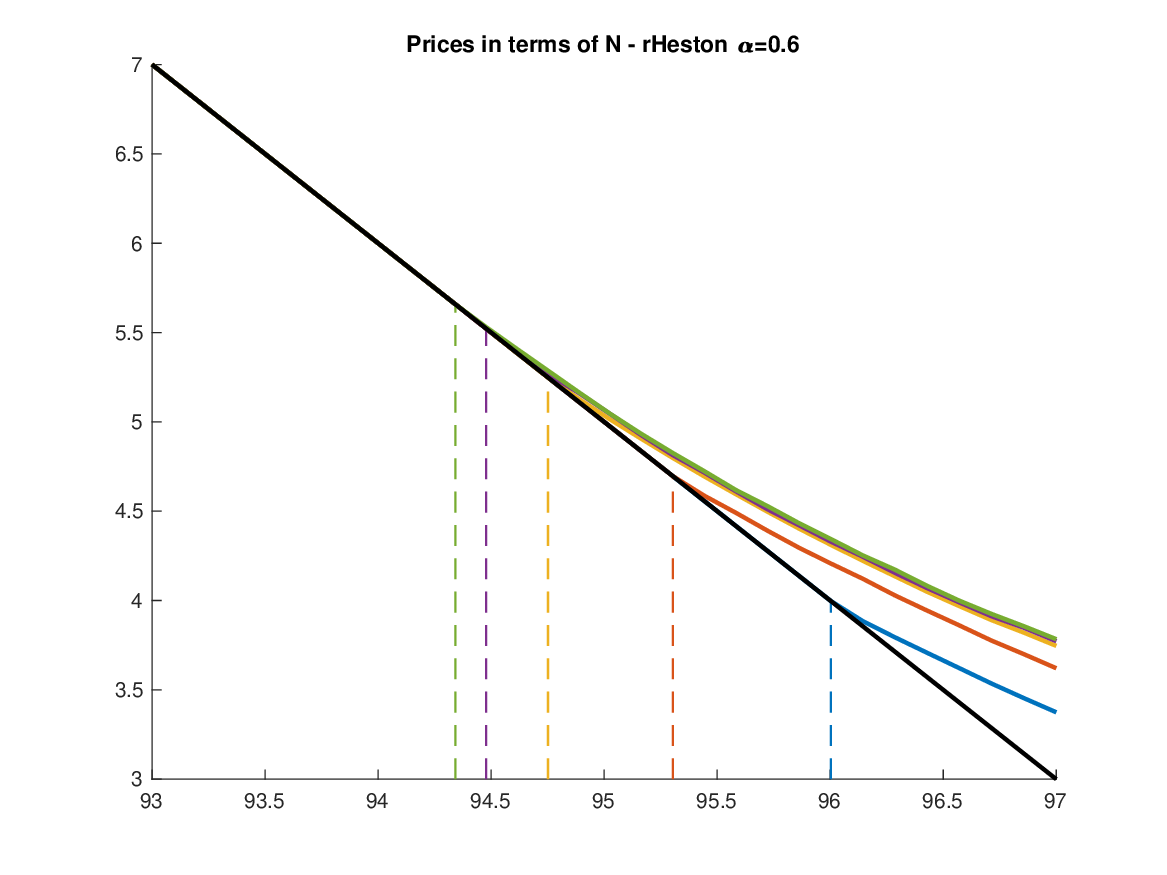}
\caption{Bermudan put option prices in terms of $N$. Payoff (black), $N=5$ (blue), $N=10$ (red), $N=25$ (yellow), $N=50$ (purple), $N=100$ (green).}\label{fig:Pricesin_N}
\end{figure}

\begin{figure}
\centering
\includegraphics[scale=0.5]{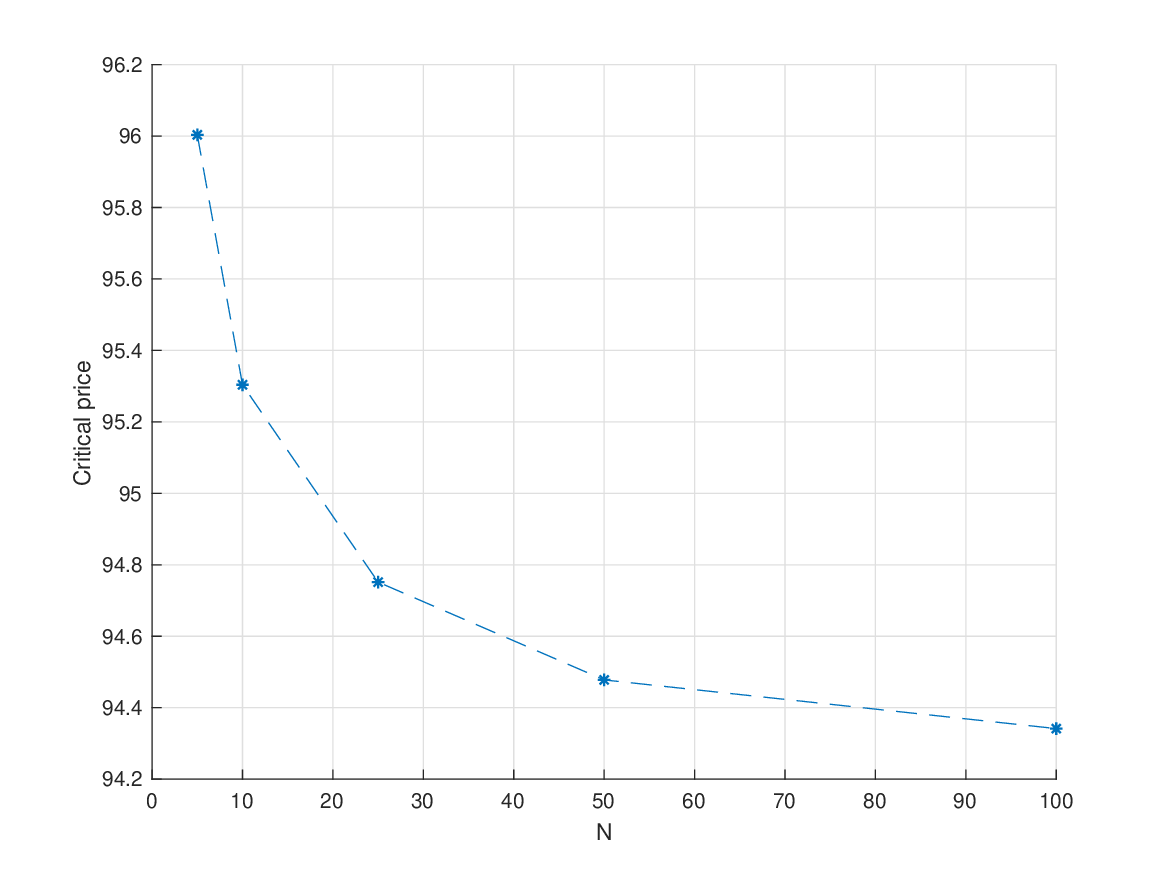}
\caption{Critical prices as a function of $N$.}\label{fig:boundaryin_N}
\end{figure}

A theoretical explanation of our numerical findings would require more detailed results about the path-behavior of the rough Heston model, and their impact on American and Bermudan option prices. Such study falls outside of the scope of this paper but could be an interesting topic of future research, along with a deeper numerical analysis of the behavior of American and Bermudan option prices in terms of the parameters of rough volatility models.

\appendix

\section{Riccati--Volterra equations}
\label{sec:app:ricatti}
\begin{proposition}\label{prop_existence_Psi}
Suppose that $K\in\Lcal^2_{loc}$ satisfies condition \ref{assump:kernelsv_0_original:1} in Assumption \ref{assump:kernelsv_0_original}. Then, given $w\in\C$ with ${\rm Re}(w)\in[0,1]$, and $h\in\Gcal_K^*$, the Riccati--Volterra equation \eqref{ricattieq} admits a solution $\Psi$ such that $\Psi(t,\cdot;w,h)\in\Gcal_K^*$ for all $t\ge 0$.
\end{proposition}
\begin{proof}
As pointed out in Remark \ref{rem:Psivspsi}, the Riccati equation \eqref{ricattieq} for $\Psi$ can be recast as the stochastic Volterra equation \eqref{eq_psi} for the function $\psi$ given by \eqref{eq:defpsi}. Thanks to the continuity of $\int_{0}^{\infty} h(\xi) K(\cdot+\xi) \dd \xi $, \cite[Theorem 12.1.1]{GLS:90} implies the existence of a continuous solution $\psi$ on a maximal interval $[0,T_{max})$. In order to prove that $T_{max}=\infty$, we can follow the proof of \cite[Lemma 7.4]{abi2019affine}. In \cite{abi2019affine} the authors consider $\Lcal^2$-solutions and a particular type of initial conditions for the Riccati--Volterra equations. In our case we consider continuous solutions and we have initial conditions of the form $\int_0^{\infty}h(\xi)K(t+\xi)\dd\xi$ such that $-\int_0^{\infty}{\rm Re}(h(\xi))K(t+\xi)\dd\xi\in\Gcal_K$. The same arguments, however, can be adapted to our setting using the invariance result in \cite[Theorem C.1]{abi2019multifactor} together with the fact that $\int_0^{\infty}f(\xi)K(t+\xi)\dd\xi\in\Gcal_K$ for all $f\in \Bcal_c(\R_+,\R_+)$. Moreover, taking the real part in \eqref{eq_psi}, \cite[Theorem C.1]{abi2019multifactor} guarantees that 
\begin{equation}\label{eq:invarianceappendix}
s\mapsto g_{t}(s)=\Delta_{t}g(s)-(\Delta_sK\ast{\rm Re}(\mathcal{R}(w,\psi)))(t)\in\Gcal_K,\quad t\ge 0,
\end{equation}
where $g(s)=-\int_0^{\infty}{\rm Re}(h(\xi))K(s+\xi)\dd \xi$. We now define $\Psi$ using \eqref{riccati_alt}, which satisfies \eqref{ricattieq} thanks to \eqref{eq_psi}. The fact that $\Psi(t,\cdot;w,h)\in\Gcal_K^*$, for all $t\ge 0$, is a consequence of \eqref{eq:invarianceappendix} and the identity
\[
\Delta_{t}g(s)-(\Delta_sK\ast{\rm Re}(\mathcal{R}(w,\psi)))(t)=-\int_0^{\infty}{\rm Re}(\Psi(t,\xi;w,h))K(s+\xi)\dd \xi.
\]
This concludes the proof.
\end{proof}

We finish this section with a sketch of the proof of Lemma \ref{lem:conv_psi}.

 \begin{proof}[Proof of Lemma \ref{lem:conv_psi}]
To prove the convergence of $\psi^n$ towards $\psi$ in $\Ccal[0,T]$, one can use similar arguments as in the proof of \cite[Theorem 4.1]{abi2019multifactor}, replacing the zero initial condition by the initial curves $\int_0^{\infty}h(\xi)K(t+\xi)\dd\xi$ and $\int_0^{\infty}h^n(\xi)K^n(t+\xi)\dd\xi$, $n\ge 1$. The convergence of $\Psi^n$ towards $\Psi$ is a consequence of the identity \eqref{riccati_alt}, the convergence of $(h^n,\psi^n)$ to $(h,\psi)$, and the quadratic structure of $\mathcal{R}(w,\cdot)$. Since ${\rm supp}(h^n)\subseteq[0,M]$ for all $n\ge1$, thanks to the form of the Riccati equations satisfied by $\Psi^n$, we conclude that the support of $\Psi^n(t,\cdot;w,h^n)$ is contained in $[0,\max\{T,M\}]$ for all $n\ge1$ and $t\le T$.
 \end{proof}

\section{Some results on the kernel approximation}
\label{sec:app:kernelapprox}

In this appendix we provide sufficient conditions on the kernel approximation which ensure condition \ref{assump:kernelsv_0:1} in Assumption \ref{assump:kernelsv_0}. 
\begin{theorem}\label{thm:appendix}
Suppose that $\mu$ is a nonnegative Borel measure on $\R_+$ such that
\begin{equation}\label{eq:condmuapp}
\int_{\R_+}(1\wedge (\eps x)^{-\frac12})\mu(\dd x)\leq c(T) \eps^{\frac{\gamma-1}{2}},\quad T>0,\,\eps\leq T, \footnote{This condition was also considered also in \cite[Section 4]{abi2019markovian}.}
\end{equation}
with $\gamma\in (0,2]$ and $c:\R_+\to\R_+$ a locally bounded function. If, in addition,
\begin{equation}\label{eq:condappetapartition}
\sup_{n\ge 1}\sup_{i\in\{0,\ldots n-1\}}\frac{\eta^n_{i+1}}{\eta_i^n}<\infty
\end{equation}
then the kernels $(K^n)_{n\ge 1}$ defined in \eqref{eq:kernelsumexp}, with $
(c_i^n)_{i=1}^n$, $(x_i^n)_{i=1}^n$ given by \eqref{eq:weightscentermass}, satisfy condition \ref{assump:kernelsv_0:1} in Assumption \ref{assump:kernelsv_0} with $\gamma$ as in \eqref{eq:condmuapp}. 
\end{theorem}

To prove Theorem \ref{thm:appendix} we use the following lemma.
\begin{lemma}\label{lem:appendix}
Suppose that $\mu$ is a nonnegative Borel measure $\mu$ on $\R_+$ such that \eqref{eq:condmuapp} holds. Let $K$ be the corresponding completely monotone kernel as in \eqref{eq:bernstein}. Then $K$ satisfies condition \ref{def:Kcal:1} in Definition \ref{def:Kcal}, with the locally bounded function $2c^2$ and the same constant $\gamma$ as in \eqref{eq:condmuapp}. 
\end{lemma}

\begin{proof}
Note that
\begin{equation*}
\|K\|_{\Lcal^2(0,\eps)}\leq \int_0^\infty \|{\rm e}^{-\cdot x}\|_{\Lcal^2(0,\eps)}\mu(\dd x)=\int_0^\infty \sqrt{\frac{1-{\rm e}^{-2x\eps}}{2x}}\mu(\dd x)\leq \eps^{\frac12}\int_0^\infty (1\wedge (\eps x)^{-\frac12})\mu(\dd x).
\end{equation*}
This implies, by \eqref{eq:condmuapp}, that $\|K\|_{\Lcal^2(0,\eps)}\leq c(T)\epsilon^{\frac{\gamma}{2}}$, $\eps\le T$. A similar argument shows that $\|\Delta_{\eps}K-K\|_{\Lcal^2(0,T)}\leq c(T)\epsilon^{\frac{\gamma}{2}}$. The conclusion readily follows from these observations.
\end{proof}

\begin{proof}[Proof of Theorem \ref{thm:appendix}]
According to Lemma \ref{lem:appendix} it is enough to show that there is a locally bounded function $\tilde{c}:\R_+\to\R_+$ such that for all $n\ge 1$,
\[
\int_{\R_+}(1\wedge (\eps x)^{-\frac12})\mu^n(\dd x)\leq \tilde{c}(T) \eps^{\frac{\gamma-1}{2}},\quad T>0,\,\eps\leq T,
\]
where $\mu^n$ is a sum of Dirac measures as in \eqref{eq:sumdiract}. This is a routine verification, using the definition of $c_i^n,x_i^n$ in \eqref{eq:weightscentermass}, Jensen's inequality, and conditions \eqref{eq:condmuapp} and \eqref{eq:condappetapartition}. For the sake of brevity, we omit the details.
\end{proof}
\begin{remark}\label{rem:appendix}
Let $K$ be the fractional kernel \eqref{frackernel} and consider the geometric partition $\eta_i^n=r_n^{i-\frac{n}{2}}$, $i=0,\ldots, n$. It is easy to check that the hypotheses of Theorem \ref{thm:appendix} hold with $\gamma=2\alpha-1$ as long as $\sup_{n\ge 1} r_n<\infty.$
\end{remark}

\bibliographystyle{abbrv}
\bibliography{biblio}

\begin{thebibliography}{10}

\bibitem{abi2019lifting}
E.~Abi~Jaber.
\newblock Lifting the {H}eston model.
\newblock {\em Quantitative Finance}, 19(12):1995--2013, 2019.

\bibitem{jaber2019weak}
E.~Abi~Jaber, C.~Cuchiero, M.~Larsson, and S.~Pulido.
\newblock A weak solution theory for stochastic {V}olterra equations of
  convolution type.
\newblock {\em The Annals of Applied Probability}, 31(6):2924--2952, 2021.

\bibitem{abi2019markovian}
E.~Abi~Jaber and O.~El~Euch.
\newblock Markovian structure of the {V}olterra {H}eston model.
\newblock {\em Statistics \& Probability Letters}, 149:63--72, 2019.

\bibitem{abi2019multifactor}
E.~Abi~Jaber and O.~El~Euch.
\newblock Multifactor approximation of rough volatility models.
\newblock {\em SIAM Journal on Financial Mathematics}, 10(2):309--349, 2019.

\bibitem{abi2019affine}
E.~Abi~Jaber, M.~Larsson, and S.~Pulido.
\newblock Affine {V}olterra processes.
\newblock {\em The Annals of Applied Probability}, 29(5):3155--3200, 2019.

\bibitem{jaber2019linear}
E.~Abi~Jaber, E.~Miller, and H.~Pham.
\newblock Linear-{Q}uadratic control for a class of stochastic {V}olterra
  equations: solvability and approximation.
\newblock {\em The Annals of Applied Probability}, 31(5):2244--2274, 2021.

\bibitem{abijaberMarkowitz2020}
E.~Abi~Jaber, E.~Miller, and H.~Pham.
\newblock Markowitz portfolio selection for multivariate affine and quadratic
  {V}olterra models.
\newblock {\em SIAM Journal on Financial Mathematics}, 12(1):369--409, 2021.

\bibitem{alfonsi2021approximation}
A.~Alfonsi and A.~Kebaier.
\newblock Approximation of {S}tochastic {V}olterra {E}quations with kernels of
  completely monotone type.
\newblock {\em arXiv preprint arXiv:2102.13505}, 2021.

\bibitem{alos2007short}
E.~Al{\`o}s, J.~A. Le{\'o}n, and J.~Vives.
\newblock On the short-time behavior of the implied volatility for
  jump-diffusion models with stochastic volatility.
\newblock {\em Finance and Stochastics}, 11(4):571--589, 2007.

\bibitem{bayer2020hierarchical}
C.~Bayer, C.~Ben~Hammouda, and R.~Tempone.
\newblock Hierarchical adaptive sparse grids and quasi-{M}onte {C}arlo for
  option pricing under the rough {B}ergomi model.
\newblock {\em Quantitative Finance}, 20(9):1457--1473, 2020.

\bibitem{bayer2021makovian}
C.~Bayer and S.~Breneis.
\newblock Makovian approximations of stochastic {V}olterra equations with the
  fractional kernel.
\newblock {\em arXiv preprint arXiv:2108.05048}, 2021.

\bibitem{bayer2016pricing}
C.~Bayer, P.~Friz, and J.~Gatheral.
\newblock Pricing under rough volatility.
\newblock {\em Quantitative Finance}, 16(6):887--904, 2016.

\bibitem{bayer2020weak}
C.~Bayer, E.~J. Hall, and R.~Tempone.
\newblock Weak error rates for option pricing under the rough {B}ergomi model.
\newblock {\em arXiv preprint arXiv:2009.01219}, 2020.

\bibitem{bayer2020pricing2}
C.~Bayer, J.~Qiu, and Y.~Yao.
\newblock Pricing options under rough volatility with backward {SPDE}s.
\newblock {\em arXiv preprint arXiv:2008.01241}, 2020.

\bibitem{bayer2020pricing}
C.~Bayer, R.~Tempone, and S.~Wolfers.
\newblock Pricing {A}merican options by exercise rate optimization.
\newblock {\em Quantitative Finance}, 20(11):1749--1760, 2020.

\bibitem{becker2019deep}
S.~Becker, P.~Cheridito, and A.~Jentzen.
\newblock Deep optimal stopping.
\newblock {\em Journal of Machine Learning Research}, 20:74, 2019.

\bibitem{bennedsen2017hybrid}
M.~Bennedsen, A.~Lunde, and M.~S. Pakkanen.
\newblock Hybrid scheme for {B}rownian semistationary processes.
\newblock {\em Finance and Stochastics}, 21(4):931--965, 2017.

\bibitem{bennedsen2016decoupling}
M.~Bennedsen, A.~Lunde, and M.~S. Pakkanen.
\newblock {Decoupling the Short- and Long-Term Behavior of Stochastic
  Volatility}.
\newblock {\em Journal of Financial Econometrics}, 2021.

\bibitem{callegaro2021fast}
G.~Callegaro, M.~Grasselli, and G.~Pag{\`e}s.
\newblock Fast hybrid schemes for fractional {R}iccati equations (rough is not
  so tough).
\newblock {\em Mathematics of Operations Research}, 46(1):221--254, 2021.

\bibitem{carmona2000approximation}
P.~Carmona, L.~Coutin, and G.~Montseny.
\newblock Approximation of some {G}aussian processes.
\newblock {\em Statistical inference for stochastic processes},
  3(1-2):161--171, 2000.

\bibitem{comte2012affine}
F.~Comte, L.~Coutin, and E.~Renault.
\newblock Affine fractional stochastic volatility models.
\newblock {\em Annals of Finance}, 8(2-3):337--378, 2012.

\bibitem{coutin1998fractional}
L.~Coutin and P.~Carmona.
\newblock Fractional {B}rownian motion and the {M}arkov property.
\newblock {\em Electronic Communications in Probability}, 3:12, 1998.

\bibitem{cuchiero2020generalized}
C.~Cuchiero and J.~Teichmann.
\newblock Generalized {F}eller processes and {M}arkovian lifts of stochastic
  {V}olterra processes: the affine case.
\newblock {\em Journal of Evolution Equations}, pages 1--48, 2020.

\bibitem{el2018microstructural}
O.~El~Euch, M.~Fukasawa, and M.~Rosenbaum.
\newblock The microstructural foundations of leverage effect and rough
  volatility.
\newblock {\em Finance and Stochastics}, 22(2):241--280, 2018.

\bibitem{el2018perfect}
O.~El~Euch and M.~Rosenbaum.
\newblock Perfect hedging in rough {H}eston models.
\newblock {\em The Annals of Applied Probability}, 28(6):3813--3856, 2018.

\bibitem{el2019characteristic}
O.~El~Euch and M.~Rosenbaum.
\newblock The characteristic function of rough {H}eston models.
\newblock {\em Mathematical Finance}, 29(1):3--38, 2019.

\bibitem{fouque2019optimal}
J.-P. Fouque and R.~Hu.
\newblock Optimal portfolio under fractional stochastic environment.
\newblock {\em Mathematical Finance}, 29(3):697--734, 2019.

\bibitem{fukasawa2017short}
M.~Fukasawa.
\newblock Short-time at-the-money skew and rough fractional volatility.
\newblock {\em Quantitative Finance}, 17(2):189--198, 2017.

\bibitem{fukasawa2019volatility}
M.~Fukasawa, T.~Takabatake, and R.~Westphal.
\newblock Is volatility rough?
\newblock {\em arXiv preprint arXiv:1905.04852}, 2019.

\bibitem{gatheral2018volatility}
J.~Gatheral, T.~Jaisson, and M.~Rosenbaum.
\newblock Volatility is rough.
\newblock {\em Quantitative Finance}, 18(6):933--949, 2018.

\bibitem{gatheral2019affine}
J.~Gatheral and M.~Keller-Ressel.
\newblock Affine forward variance models.
\newblock {\em Finance and Stochastics}, 23(3):501--533, 2019.

\bibitem{goudenege2020machine}
L.~Gouden{\`e}ge, A.~Molent, and A.~Zanette.
\newblock Machine learning for pricing {A}merican options in high-dimensional
  markovian and non-markovian models.
\newblock {\em Quantitative Finance}, 20(4):573--591, 2020.

\bibitem{GLS:90}
G.~Gripenberg, S.-O. Londen, and O.~Staffans.
\newblock {\em Volterra integral and functional equations}, volume~34 of {\em
  Encyclopedia of Mathematics and its Applications}.
\newblock Cambridge University Press, Cambridge, 1990.

\bibitem{guennoun2018asymptotic}
H.~Guennoun, A.~Jacquier, P.~Roome, and F.~Shi.
\newblock Asymptotic behavior of the fractional {H}eston model.
\newblock {\em SIAM Journal on Financial Mathematics}, 9(3):1017--1045, 2018.

\bibitem{han2020mean}
B.~Han and H.~Y. Wong.
\newblock Mean--variance portfolio selection under {V}olterra {H}eston model.
\newblock {\em Applied Mathematics \& Optimization}, pages 1--28, 2020.

\bibitem{han2021merton}
B.~Han and H.~Y. Wong.
\newblock Merton's portfolio problem under {V}olterra {H}eston model.
\newblock {\em Finance Research Letters}, 39:101580, 2021.

\bibitem{harms2019affine}
P.~Harms and D.~Stefanovits.
\newblock Affine representations of fractional processes with applications in
  mathematical finance.
\newblock {\em Stochastic Processes and their Applications}, 129(4):1185--1228,
  2019.

\bibitem{heston1993closed}
S.~L. Heston.
\newblock A closed-form solution for options with stochastic volatility with
  applications to bond and currency options.
\newblock {\em The Review of Financial Studies}, 6(2):327--343, 1993.

\bibitem{horvath2017functional}
B.~Horvath, A.~J. Jacquier, and A.~Muguruza.
\newblock Functional central limit theorems for rough volatility.
\newblock {\em Available at SSRN 3078743}, 2017.

\bibitem{jaisson2016rough}
T.~Jaisson and M.~Rosenbaum.
\newblock Rough fractional diffusions as scaling limits of nearly unstable
  heavy tailed {H}awkes processes.
\newblock {\em The Annals of Applied Probability}, 26(5):2860--2882, 2016.

\bibitem{keller2018affine}
M.~Keller-Ressel, M.~Larsson, and S.~Pulido.
\newblock Affine rough models.
\newblock {\em arXiv preprint arXiv:1812.08486}, 2018.

\bibitem{lapeyre2019neural}
B.~Lapeyre and J.~Lelong.
\newblock Neural network regression for {B}ermudan option pricing.
\newblock {\em Monte Carlo Methods and Applications}, 27(3):227--247, 2021.

\bibitem{longstaff2001valuing}
F.~A. Longstaff and E.~S. Schwartz.
\newblock Valuing {A}merican options by simulation: a simple least-squares
  approach.
\newblock {\em The Review of Financial Studies}, 14(1):113--147, 2001.

\bibitem{P05}
P.~E. Protter.
\newblock Stochastic differential equations.
\newblock In {\em Stochastic integration and differential equations}, pages
  249--361. Springer, 2005.

\bibitem{RY:99}
D.~Revuz and M.~Yor.
\newblock {\em Continuous martingales and {B}rownian motion}, volume 293 of
  {\em Grundlehren der Mathematischen Wissenschaften [Fundamental Principles of
  Mathematical Sciences]}.
\newblock Springer-Verlag, Berlin, third edition, 1999.

\bibitem{romer2020hybrid}
S.~E. R{\o}mer.
\newblock Hybrid multifactor scheme for stochastic {V}olterra equations.
\newblock {\em Available at SSRN 3706253}, 2021.

\end{thebibliography}

\end{document}